\documentclass[10pt, twocolumn]{IEEEtran}

\usepackage{amsmath}
\usepackage{amssymb}    
\usepackage{amsfonts}
\usepackage[english]{babel}
\usepackage{cite} 
\usepackage{multirow}
\usepackage{stfloats}    

\usepackage{algorithm}   
\usepackage{algorithmic} 

\usepackage{caption} 
\usepackage{graphicx}
\usepackage{epsfig}
\usepackage{subfigure} 
\usepackage{tablefootnote}

\usepackage{url}

\usepackage{accents}
\makeatletter
\def\widebar{\accentset{{\cc@style\underline{\mskip10mu}}}}
\def\Widebar{\accentset{{\cc@style\underline{\mskip13mu}}}}
\makeatother

\newtheorem{theorem}{Theorem}

\newtheorem{definition}{Definition}
\newtheorem{proposition}{Proposition}
\newtheorem{corollary}{Corollary}
\newtheorem{example}{Example}

\usepackage{booktabs}    
 \usepackage[draft]{changes}   

\begin{document}

\captionsetup[figure]{labelfont={ }, name={Fig.}, labelsep=period} 
\pagestyle{empty}

\title{ Age-upon-Decisions Minimizing  Scheduling in Internet of Things: To be Random
            \\or to be Deterministic?}
\author{Yunquan~Dong,~\IEEEmembership{Member,~IEEE},
        Zhengchuan~Chen,~\IEEEmembership{Member,~IEEE},  \\
        Shanyun~Liu,~Pingyi Fan,~\IEEEmembership{Senior Member,~IEEE},
        and~Khaled~Ben~Letaief,~\IEEEmembership{Fellow,~IEEE}

        \thanks{

                        Y. Dong is with the School of Electronic and Information Engineering, the Jiangsu Collaborative Innovation Center of Atmospheric Environment and Equipment Technology (CICAEET), Nanjing University of Information Science and Technology, Nanjing 210044, China (e-mail: yunquandong@nuist.edu.cn).

                        Z. Chen is with the College of Microelectronics and Communication Engineering, Chongqing University, Chongqing 400044, China (e-mail: czc@cqu.edu.cn).

                         S. Liu and P. Fan are with the Department of Electronic Engineering, the Beijing National Research Center for Information Science and Technology, Tsinghua University, Beijing 100084, China. (e-mail: liushany16@mails.tsinghua.edu.cn, fpy@tsinghua.edu.cn).

                          Khaled B. Letaief is with the Department of Electrical and Computer Engineering, HKUST, Clear Water Bay, Kowloon, Hong Kong (e-mail: eekhaled@ust.hk).
        }
        }


\maketitle
\thispagestyle{empty}

\vspace{-15mm}
\begin{abstract}
We consider an Internet of Things (IoT) system in which a sensor delivers updates to a monitor with exponential service time and first-come-first-served (FCFS) discipline.
    We investigate the freshness of the received updates and propose a new metric termed as \textit{Age upon Decisions (AuD)}, which is defined as the time elapsed from the generation of each update to the epoch it is used to make decisions (e.g., estimations, inferences, controls).
Within this framework, we aim at improving the freshness of updates at decision epochs by scheduling the update arrival process and the decision making process.
    Theoretical results show that 1) when the decisions are made according to a Poisson process, the average AuD  is independent of decision rate and would be minimized if the arrival process is periodic (i.e., deterministic);
2) when both the decision process and the arrive process are periodic, the average AuD is larger than, but decreases with decision rate to, the average AuD of the corresponding system with Poisson decisions (i.e., random);
    3)  when both the decision process and the arrive process are periodic, the average AuD can be further decreased by optimally controlling the offset between the two processes.
For practical IoT systems, therefore, it is suggested to employ periodic arrival processes and random decision processes.
    Nevertheless, making periodical updates and decisions with properly controlled offset also is a promising solution if the timing information of the two processes can be accessed by the monitor.
\end{abstract}


\section{Introduction}
\IEEEPARstart {R}{cent} developments in wireless sensor networks, embedded systems, and low power communications have made the Internet of Things (IoT) a reality.
    With a huge and increasing number of smart devices connected to the internet, IoT networks are increasingly popular in various scenarios related to data gathering and service sharing.
Typical applications of IoT include environment monitoring~\cite{Monitoring-IoTJ-2018}, IoT based mobile phone computing systems~\cite{Mobile-IoTJ-2016}, Industrial IoT systems~\cite{IIoT-IoTJ-2018}, and so on.
    Facilitated with the IoT technology, our world becomes smarter and smarter.

In particular, IoT has spawned a lot of applications with stringent delay requirements.
     In smart vehicular networks, for example, vehicles need to share their status (e.g., position, speed, acceleration) with each other timely to ensure safety~\cite{Vnet-1-2011}.
In IoT-based smart sensing systems (e.g., environmental monitoring, precision agriculture, disaster and emergency response), the sensing data should be exploited in a real/near-real-time manner~\cite{remote-Proc-2017}.
    In health monitoring (e.g., heart failure detection) systems \cite{health-Proc-2012}, asset tracking applications \cite{asset-PMC-2016}, and the indoor-positioning systems \cite{remote-Proc-2017}, timely information acquiring is also crucial.
    Different from traditional systems where communication takes place  to reproduce the messages of the source,
     the delivered information and updates are used \textit{to control}, or \textit{to compute}, or \textit{to infer} in these systems~\cite{tutorial-2016}.
Therefore, the freshness (the timeliness) of the received information/updates is of crucial importance.

To this end, \textit{Age of Information (AoI)} was proposed as a performance metric of information freshness~\cite{Yates-2012-age, Vnet-2-2011}.
    Specifically, AoI is defined as the elapsed time since the generation of the latest received update, i.e., the age of the latest update.
\textit{First}, AoI characterizes the freshness of received updates more precisely than traditional measures like delay and throughput.
    For example, when transmission delay is small, the received data is only fresh at the time it is received and becomes less fresh as the time approaches the next data reception, especially when throughput is low.
When throughput is large, the received data would also be not fresh if they had undergone large queueing delays at the transmitter.
    \textit{Second}, the AoI metric enables direct comparison between lossless and lossy systems \cite{SHS-2018}.
Since increasing update rate would induce larger recovery delay in lossless systems while increase recovery distortion or packet dropping in lossy systems, the delay comparison and the throughput comparison between lossy and lossless systems are difficult to interpret.
    With the AoI measure, however, lossy systems and lossless systems are comparable since AoI is independent of packet loss.

\subsection{Motivations}
With the following observations that
\begin{itemize}
  \item \textit{delay} quantifies the freshness of updates at epochs when they are received;
  \item \textit{AoI} quantifies the freshness of updates at every epoch after they are received;
  \item in many IoT systems, the freshness of updates are only important at some decision epochs,
\end{itemize}
however, we are motivated to consider a new freshness measure termed as \textit{Age upon Decisions (AuD)}.
    To be specific,
\begin{itemize}
  \item \textit{AuD} quantifies the freshness of updates at those decision epochs when the received updates are used to make decisions.
\end{itemize}

By making a \textit{decision}, we mean that the received update is used by the monitor to facilitate subsequent actions, e.g., to control,  to compute,  to infer, or to decode.
    In this sense, the monitor can be treated as a \textit{Decision Making Unit (DMU)} while the corresponding IoT system can be referred to as an \textit{update-and-decide system}.
In particular, we denote general processes with $G$, deterministic processes with $D$, and Poisson processes with $M$.
    For example, in a $D/G/1/M$ update-and-decide system, the arrival process is a deterministic (periodic) process, the service process is a general one, and the decision process is a Poisson process.

\begin{example}
    In IoT based smart agriculture systems, the sensors transmit their observations on plants, soil, and air condition to a monitor with random channel access.
        After having collected enough data, the monitor evaluates the status of the farmland and notifies the manager with an interim report.
    Since the unexpected events in agriculture are highly unpredictable, fast and proper reactions which can only be performed based on real-time monitoring and immediate reporting, are important to prevent pets and plant diseases \cite{SmartAgriculture-2016}.
The evaluating process (decision making process) at the monitor, however, can never be consistent with the observing process (update arrival process) and the transmission process (service process) of sensors.
    Thus, AuD would be suitable to evaluate the timeliness of interim reports.
In particular, by scheduling the observing process of each sensor and the evaluating process at the monitor, the freshness of the reports can be guaranteed. $\hfill{} \blacksquare$
\end{example}

\begin{example}
    In modern intelligent transportation systems where live traffic information is available to each driver, dynamic route planning is crucial for traffic load balancing and congestion section avoiding.
        To be specific, by timely reporting the global positioning system (GPS) information of cars and the pictures taken by street cameras to the map operator, the traffic load can be evaluated and any emergent accidents can be detected.
    Based on these information, each vehicle can then plan its route dynamically, i.e., determine whether to make a turn at the next cross in real time.
        Since the car-position process (update arrival process), the information collecting process (service process), and route planning process (the decision process) are all random, AuD would be a proper measure for the freshness of dynamic route planning.
    Moreover,  the car-position process and the route planning process can be controlled and optimized.
    $\hfill{} \blacksquare$
\end{example}

\begin{example}
    In health monitoring systems, timely detecting/reporting heart failures is one of the most important objectives \cite{health-Proc-2012}.
        Based on the acceleration information collected by the smartphone accelerometer, the activity of a patient can be recognized, which is important to determine whether the patient is in the normal state or suffers heart failure.
    Although increasing the rate of accelerometer reading could ensure the timeliness of detecting potential heart failure, it also increases the energy consumption of the smartphone.
        Therefore, it is natural to determine the reading rate based on the level of illness and schedule the readings (arrivals, deterministic or random) and medical staff observations (decisions, deterministic or random) properly.
$\hfill{} \blacksquare$
\end{example}

    In this paper, therefore, we shall investigate such a fundamental problem: \textit{what are the average-AuD-minimizing scheduling for the update arrival process and the decision making process in update-and-decide systems}.
To be specific, we are interested in the optimal control of inter-arrival times and inter-decision times.
    In doing so, we shall  provide answer to whether the arrival process and the decision process should be deterministic or be random.

\subsection{Main Contributions}
In this paper, we first consider a $G/G/1/M$ IoT based update-and-decide system and show that the  average AuD is independent of the rate of decisions.
    Thus, we can focus on the arrival process and study which type of arrival process minimizes the average AuD.
We then consider a $D/M/1/D$ system in which both the arrival process and the decision process are periodic.
    In particular, we investigate the average AuD of the system and the average-AuD-optimal offset between the arrival process and the decision process.
The main contribution of the paper can be summarized as follows.
\begin{itemize}
    \item \textit{Novel freshness measure:}  We propose an \textit{AuD} measure to characterize the freshness of updates at those important epochs, i.e., the decision epochs.
      \item \textit{Optimal update and decision scheduling:} We show that the periodic (i.e., deterministic) arrival process minimizes the average AuD of $G/M/1/M$ update-and-decide systems.
            For the decision process of $D/M/1/G$ systems, however, a synchronous and periodic one yields larger average AuD than a Poisson process (i.e., random).
        Nevertheless, by optimizing the offset between the arrivals and the decisions of an asynchronous $D/M/1/D$ system, the average AuD can be smaller than that of a $D/M/1/M$ system.
    \item \textit{Efficient algorithms:} We present an  algorithm to search the optimal arrival process for $G/M/1/M$ systems, an algorithm to search the key parameter $\rho_1$ for $G/G/1/M$ systems, and an algorithm to optimally control the timing offset in $D/M/1/D$ systems.
            These iterative algorithms are efficient to converge in a few iterations in the discussed examples.
\end{itemize}

\subsection{Organization}
The rest of the paper is organized as follows.
In  Section \ref{sec:2_related_work}, several related works are reviewed.
   In Section~\ref{sec:2_model}, we present the update-and-decide system model and the definition of AuD.
 We then investigate the  average AuD of $G/G/1/M$ update-and-decide systems with random update arrivals in Section~\ref{sec:3_aud_rand}.
    For $G/M/1/M$ update-and-decide systems, we also present an efficient algorithm searching the optimal distribution parameters for inter-arrival time.
 In Section~\ref{sec:4_aud_determine}, we investigate the average AuD of $D/M/1/D$ update-and-decide systems with periodic arrivals and periodic decisions.
    For the case where the arrival rate and the decision rate is equal, we further present the average-AuD-minimizing control for the offset between the arrival process and the decision process.
Finally,  our work is concluded in Section~\ref{sec:5_conclusion}.

\section{Related Work} \label{sec:2_related_work}
As indicated by the definition, AoI specifies the age of the latest received update~\cite{Yates-2012-age}.
    Thus, AoI is an absolute measure that is comparable among different systems, e.g., systems with lossless or lossy communications, and systems with different applications.
For example, AoI has been extensively studied in various queueing systems, e.g., $M/M/1, M/D/1$ and $D/M/1$~\cite{Yates-2012-age}, and under several serving disciplines, e.g., first-come-first-served (FCFS)~\cite{Yates-2012-age}, last-generate-first-served (LGFS)~\cite{Sun-2016-mlt-sver}.
    In general, it is very challenging to investigate the statistics of AoI with queueing theory, except a few successful attempts on the distribution of AoI, e.g., the AoI distibution for single server queues in~\cite{Inoue-2017-distribution}.
To this end,  the stochastic hybrid systems (SHS) theory was introduced as a key tool for AoI analysis \cite{SHS-2018}.
    With the SHS approach, the authors have analyzed the temporal convergence of higher order moments and  moment generating function (MGF) in a class of status sampling networks~\cite{Inoue-2017-distribution}.

Due to its specialty in characterizing the timeliness of updates and transmissions, AoI is closely relative to various real-time scenarios and has been widely applied to sensor-based monitoring \cite{Hr-2019-correlated, Gu-20197-mornitoring, Niu-2019-RR1}, cognitive radio-based IoT systems \cite{Gu-2019-cognitive}, and two-way data exchanging systems \cite{Dong-2019-access, Dong-2019-jcn}.
    In \cite{Hr-2019-correlated}, the authors proposed to minimize the average AoI of updates and increase the life-time of mobile devices by transmitting some assisting updates from a correlated source.
In multi-terminal based monitoring systems, the age-energy trade-off and link-layer retransmission schemes were investigated in \cite{Gu-20197-mornitoring} while the multiple-access-layer load balancing schemes (e.g., round-robin) was investigated in \cite{Niu-2019-RR1}.
    For cognitive radio-based IoT networks, the critical update rate optimizing the primary system was obtained asymptotically in \cite{Gu-2019-cognitive}, which provided solid foundation to determine whether the overlay scheme outperforms the underlay scheme or not.
Moreover, for two-way data exchanging systems with wireless power transfer at the master node, the achievable uplink-downlink timeliness region of time-splitting systems was presented in \cite{Dong-2019-access} and the weighted-sum-AoI optimal power splitting  (between energy flow and information flow) scheme was considered in \cite{Dong-2019-jcn}.

There were also new freshness measures proposed for some specific systems and applications~\cite{kosta-AoI-2017, Effective-AoI-2018, Mutual-AoI-2018}.
    In ~\cite{kosta-AoI-2017}, value of information updates (VoIU) was proposed to measure the reduction in delay cost upon the reception of a new update.
In~\cite{Effective-AoI-2018},  the authors investigated the connection between information age and what they called effective age, which is closely related with the structure information and the pattern of sampling, and is minimized when the prediction error is minimized.
    Furthermore, the mutual information between received samples and  source signals was used to measure the freshness of received samples in~\cite{Mutual-AoI-2018}.

With the observation that AoI  specifies the age of updates at arbitrary epochs, we proposed to emphasize the age of updates at decision epochs using the AuD measure in our previous work~\cite{Dong-AuD-2018}.
    In the paper, the average AuD of an $M/M/1/M$ update-and-decide system was obtained explicitly.
The result was then extended to general $G/G/1/M$ update-and-decide systems in~\cite{Dong-AuD-G-2018}.
    In this paper, we shall further consider the performance of $D/M/1/D$ systems and minimize the average AuD of various update-and-decide systems by scheduling the arrival process and the decision process.

\section{System Model}\label{sec:2_model}

We consider an IoT based update-and-decide system with arrival rate $\lambda$ and service rate $\mu$, as shown in Fig. \ref{fig:net_model}.
    The arrived updates are stored in an infinite long buffer and will be served according to the FCFS discipline.
Based on the received updates, a monitor (DMU) makes random decisions at rate $\nu$.

As shown in Fig. \ref{fig:aud}, the updates are generated at \textit{arrival epochs} $\{t_k, k=1,2,\cdots\}$ and are received by the DMU at \textit{departure epochs} $t'_k$.
    The \textit{inter-arrival time} $X_k$ between neighboring updates is $X_{k}=t_k-t_{k-1}$ and the \textit{system time} that packet $k$ stays in the system is $T_{k}=t'_k-t_k$.
Note that  system time $T_k$  is the sum of \textit{waiting time}  $W_{k}$ and  service time $S_{k}$,  i.e., $T_{k} =W_{k}+S_{k}$.
    We refer to the period between two consecutive departure epochs as \textit{inter-departure time} $Y_k=t_k'-t_{k-1}'$ and the period between two consecutive decision epochs as \textit{inter-decision time} $Z_j=\tau_j-\tau_{j-1}$.

On the update-and-decide system, we consider a set of the following assumptions.
\begin{itemize}
  \item [A1] The arrival rate is smaller than the service rate (i.e., $0<\lambda<\mu$) so that the system is stable.
  \item [A2] The inter-decision time $Z_j$ is exponentially distributed with probability distribution function (pdf) $f_\text{Z}(x)=\nu e^{-\nu x}$ for $x\geq0$, unless otherwise stated.
        That is, the decisions are made following a Poisson Process and the system can be denoted by $G/G/1/M$.
\end{itemize}

In this paper, we investigate the freshness of the received updates at  decision epochs via age upon decisions.

\begin{definition}
     (\textit{Age upon Decisions-AuD}). At the $j$-th decision epoch, denote the index of the most recently received update as
          \begin{equation} \nonumber
                N_\text{U}(\tau_j) = \max\{ k|t_k'\leq \tau_j \},
            \end{equation}
      and the generation time of the update as
          \begin{equation} \nonumber
                U(\tau_j) = t_{N_\text{U}(\tau_j)}.
            \end{equation}

      The AuD of the update-and-decide system is then defined as the random process
    \begin{equation} \label{df:aud}
        \Delta_\text{D}(\tau_j) = \tau_j - U(\tau_j).
    \end{equation}
\end{definition}

That is, $\Delta_\text{D}(\tau_j)$ characterizes the freshness of update $N_\text{U}(\tau_j)$ at the epoches it is used to make decisions.
    Note that if we replace  decision epoch $\tau_j$ with arbitrary time $t$, AuD $\Delta_\text{D}(\tau_j) $  would reduce to AoI $\Delta(t) $.

\begin{figure}[!t]
\centering
\includegraphics[width=2.5in]{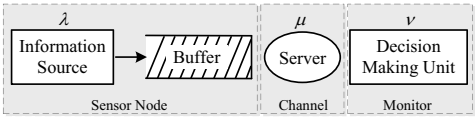}
\caption{The update-and-decide system model. } \label{fig:net_model}
\end{figure}

\begin{figure}[!t]
\centering
\includegraphics[width=2.9in]{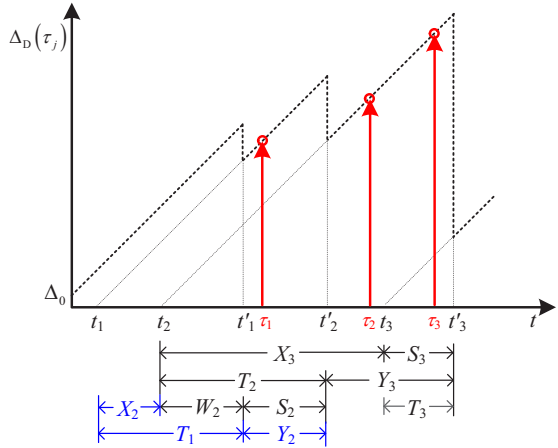}
\caption{Age upon decisions, where $X_{k}=t_k-t_{k-1}$ is the inter-arrival time, $S_k$ is the service time, $W_k$ is the waiting time, $T_{k}=t'_k-t_{k}$ is the system time, $Y_{k}=t'_k-t'_{k-1}$ is the inter-departure time, and $Z_{j}=\tau_j-\tau_{j-1}$ is the inter-decision time.  } \label{fig:aud}
\end{figure}

\begin{example}\label{eg_4}
    Fig.~\ref{fig:aud} shows a sample path of AoI and AuD, where  AoI is presented by dotted line segments and  AuD is presented by empty circles.
        Since the service of the first update is not completed until $t_1'$, the second update sees a busy server upon its arrival at $t_2$.
  The second update waits for a period of $W_2$ and gets served immediately at the departure of the first update.
Thus,  inter-departure time $Y_2$ between the first and the second updates  equals to the service time of the second update, i.e., $Y_2=S_2$.
     This is a typical case in which $X_k<T_{k-1}$ is true and we have $Y_k=S_k$.
  On the other hand, if $X_k>T_{k-1}$ is true (e.g., $X_3>T_2$), the next update has not arrived at the departure of update $k$.
    As shown in Fig.~\ref{fig:aud}, the server will be idle for a period of $X_3-T_2$ before the third update gets served from its arrival.
  In this case, the inter-departure time is $Y_k=X_k+S_k-T_{k-1}$.

  During each inter-departure time, several decisions can be made based on the latest update, i.e., the former update is used for several times.
        For example, there are two decision epochs ($\tau_2$ and $\tau_3$, denoted by the red arrow) during $Y_3$ and the corresponding AuD are $\Delta_\text{D}(\tau_2)$ and $\Delta_\text{D}(\tau_3)$, respectively.
  There may also be cases where no decision is made during an inter-departure time, which means that the former update is not used.
        $\hfill{} \blacksquare$
\end{example}

For the given arrival process, the serving process, and the decision process, we are interested in the \textit{average AuD} of the system.
    Suppose there are $N_T$ decisions during a period of $T$,  average AuD is given by
    \begin{equation} \label{df:aud}
        \widebar{\Delta}_\text{D} =\lim_{T\rightarrow\infty}\frac{1}{N_T} \sum_{j=1}^{N_T} \Delta_\text{D}(\tau_j),
    \end{equation}
    with $\lim_{j\rightarrow\infty}\tau_j=+\infty$.

\section{Average AuD with Random Arrivals} \label{sec:3_aud_rand}
    In this section, we investigate the average AuD of update-and-decide systems with random arrivals and general/exponential services.

\subsection{AuD of  G/G/1/M Update-and-Decide Systems} \label{subsuc:2_A}

We denote $X_1=t_1$. 
Since  departure time $t_k'$ can be expressed as
\begin{equation*}
    t_k' =\sum_{i=1}^k X_i +T_k,
\end{equation*}
      the inter-departure time $Y_k=t_k'-t_{k-1}'$ can be rewritten as
\begin{equation*}
    Y_k = X_k +T_k - T_{k-1},\qquad k\geq1.
\end{equation*}

By considering the number of decisions made during each inter-departure time, the average AuD of the system can be obtained, as shown in the following theorem.
\begin{theorem}\label{thm:thm_1_aud_g_1}
    In a $G/G/1/M$ update-and-decide system with general arrival process, general service process,  and Poisson decision process,  the average AuD of the system is independent of  decision rate and is given by
    \begin{equation}\label{rt:thm_aud_g_1}
        \widebar{\Delta}_\text{D} = \frac{\mathbb{E}[Y_k^2]+2\mathbb{E}[T_{k-1}Y_k]}{2\mathbb{E}[Y_k]}.
    \end{equation}
\end{theorem}

\begin{proof}
    See Appendix \ref{apx:prf_1_aud_g_1}.
\end{proof}

From \textit{Theorem} \ref{thm:thm_1_aud_g_1}, we have the following observations.

\begin{itemize}
  \item[i)] Making decisions more frequently does not improve the timeliness of decisions.
                This is because when the decision rate is increased, although there will be more decision epochs being closer to departure epochs, there will also be more decision epochs being farther from departure epochs.
            In the statistical sense, therefore, the average AuD does not change with the rate of decisions.
  \item[ii)] With Poisson distributed decisions, the average AuD is equal to the average AoI of the system.
            Note that the average AuD converges to the average AoI as the decision rate goes to infinity.
                Together with \textit{Theorem} \ref{thm:thm_1_aud_g_1}, it is readily seen that the average AuD and the average AoI would be equal regardless of decision rate.
            Therefore, the AuD framework also provides an alternative approach of calculating average AoI.
  \item[iii)] In the AuD framework, we have a new dimension of optimization, i.e., the decision process.
            In addition to optimizing the arrival process and the service process, therefore, we can further reduce average AuD by scheduling the decision process.
  \item[iv)] Although the average AuD is independent of decision rate, it can  be reduced by scheduling the arrival process and the service process.
            For example, scheduling the arrival processes by varying the distribution of inter-arrival time.
\end{itemize}

\subsection{AuD of G/M/1/M Update-and-Decide Systems} \label{subsuc:2_A}
In this subsection, we consider a $G/M/1/M$ update-and-decide system where the service time is exponentially distributed with mean service time $1/\mu$ and the inter-arrival time is generally distributed with mean  $1/\lambda$.
    We  denote the pdf of  inter-arrival time $X_k$ as $f_\text{X}(x)$ and the pdf of service time $S_k$  as $f_\text{S}(x)$, i.e.,  $f_\text{S}(x)=\mu e^{-\mu x}$.

\vspace{1mm}
\noindent~\textit{1) Queueing Analysis}
\vspace{1mm}

Since $f_\text{X}(x)$ does not have the memoryless property as the exponential distribution, the remaining time to the next arrival depends on the passed time after the previous arrival.
    The length of the update queue, therefore, is not a Markov process in general.
To investigate the stationary distribution of queue length, we need to consider an embedded Markov chain within the $G/M/1$ queue first.
    In particular, the embedded time instants are exactly the time of update arrivals.

We denote the number of updates in the system just prior the $k$-th arrival as $L_\text{a}(k)$, which take values from state space $\Omega=\{0,1,\cdots,\}$.
    As shown in \cite[Chap.\,14.\,8]{book_queueing_2006}, $L_\text{a}(k)$ is a Markov chain with stationary distribution $\boldsymbol{\pi}$, where
\begin{equation}\label{rt:pi_L}
    \pi_j = (1-\rho_1)\rho_1^j  ~~ \text{for}~j\geq0.
\end{equation}
Moreover, $\rho_1$ satisfies
\begin{equation} \label{df:rho_a}
    \rho_1 = \int_0^\infty f_\text{X}(x) e^{-\mu(1-\rho_1)x} \text{d}x.
\end{equation}
    Therefore, $\rho_1$  can be solved iteratively using \textit{Algorithm} \ref{alg:rho_a}.

\begin{algorithm}[!t]
\caption{Iterative Solution of $\rho_1$}
\begin{algorithmic}[1]\label{alg:rho_a}
\REQUIRE ~~\\
    \STATE Set $\rho_1^{(0)}$ close to 1, $\Delta\rho_1$ sufficiently large,  and $\epsilon$  reasonably small;
    \STATE Set $i=0$;
\ENSURE ~~\\
\WHILE{$|\Delta\rho_1|>\epsilon$}
\STATE $\rho_1^{(i+1)}=\int_0^\infty f_\text{X}(x) e^{-\mu(1-\rho^{(i)}_1)x} \text{d}x$;
\STATE $ \Delta\rho_1 =\rho_1^{(i+1)}-\rho_1^{(i)}$;
\STATE $i=i+1$;
\ENDWHILE
 \STATE \textbf{Output:} $\rho_1=\rho_1^{(i)}$
\end{algorithmic}

\end{algorithm}

Since system time $T_{k}$ is only dependent with the queue length at  arrival epochs, the pdf of $T_k$ can be readily obtained based on $\boldsymbol{\pi}$, as shown in the following proposition.
\begin{proposition} \label{prop:fx_Tk}
    Given that $S_k$ is exponentially distributed, $T_k$ is exponentially distributed with pdf
    \begin{equation*}
        f_\text{T}(x)= \mu(1-\rho_1) e^{-\mu(1-\rho_1)x}, \quad x\geq0.
    \end{equation*}
\end{proposition}

\textit{Proposition} \ref{prop:fx_Tk} is important in that it ensures further analysis on  inter-departure time and  average AuD possible.
    In particular, we have
    \begin{eqnarray*}
        \hspace{-6mm}&&\mathbb{E}[T_{k}] = \frac{1}{\mu(1-\rho_1)},\\ \label{rt_rho1}
        \hspace{-6mm}&&\Pr\{X_k<T_{k-1}\} =  \rho_1.
    \end{eqnarray*}

We also have the following proposition on inter-departure time $Y_k$.
\begin{proposition}\label{prop:moments_Yk}
    The first two order moments of inter-departure time $Y_k$ and the cross correlation between $T_{k-1}$ and $Y_k$ are, respectively, given by
    \begin{eqnarray*}
        \mathbb{E}[Y_k] \hspace{-2.74mm}&=&\hspace{-2.75mm}  \mathbb{E}[X_k],\\
        \mathbb{E}[Y_k^2] \hspace{-2.74mm}&=&\hspace{-2.75mm} \mathbb{E}[X_k^2] - \frac{2\rho_1}{\mu(1-\rho_1)}\mathbb{E}[X_k] + \frac{2}{\mu^2(1-\rho_1)}, \\
        \mathbb{E}[T_{k-1}Y_k] \hspace{-2.74mm}&=&\hspace{-2.75mm}  \frac{1}{\mu(1-\rho_1)}\mathbb{E}[X_k] - \frac{1}{\mu^2(1-\rho_1)} + \frac{q_1}{\mu(1-\rho_1)},
    \end{eqnarray*}
    where
    \begin{equation*}
        q_1=\int_0^\infty xf_\text{X}(x)e^{-\mu(1-\rho_1)x} \text{d}x.
        \end{equation*}
\end{proposition}

\vspace{1mm}
\noindent~\textit{2) AuD-Optimal Arrivals}
\vspace{1mm}

Combining the results in \textit{Theorem} \ref{thm:thm_1_aud_g_1},  \textit{Proposition} \ref{prop:fx_Tk}, and \textit{Proposition} \ref{prop:moments_Yk},  the average AuD can readily be expressed.
    We can then minimize the average AuD and determine the optimal distribution of  inter-arrival time by considering the following functional optimization problem.

\begin{align}
\label{prob:1_0}
 h_\text{opt}=\min\limits_{f_\text{X}(x)}~~ & \frac{\mathbb{E}[Y_k^2]+2\mathbb{E}[T_{k-1}Y_k] }{2\mathbb{E}[Y_k]} \qquad\qquad\qquad\\
 \label{prob:1_1}
  (\textbf{P}_1)~~~~~~~~~~~~\text{s.~t.}~~  &   \int_0^\infty f_\text{X}(x) \text{d}x=1,\\
\label{prob:1_2}
            & f_\text{X}(x)\geq0,\quad \forall x\geq 0, \\
 \label{prob:1_3}
            &  \hspace{-2mm} \rho < 1,
\end{align}
where $\rho=\lambda/\mu$ and  constraint \eqref{prob:1_3} guarantees that the stationary distribution $\boldsymbol{\pi}$ (see \eqref{rt:pi_L}) exists and the queue is stable.

    As shown in \cite{Sun-2016-zero-wait,Boyd-2004-convex}, however, the objective function in Problem \textbf{P}$_1$ is quasi-convex.
In general, it is difficult to solve Problem $\textbf{P}_1$ and determine which type of probability distribution minimizes the average AuD.
    For a certain specified type of distribution, however, Problem $\textbf{P}_1$ can be solved by considering the feasibility of the following functional problem: 
\begin{eqnarray*}
         \text{find}  \hspace{-3mm}&&  \hspace{-2mm} f_\text{X}(x) \\
  (\textbf{P}_2)~~~~~\text{s.~t.}  \hspace{-3mm}&&  \hspace{-2mm}\mathbb{E}[Y_k^2] +2\mathbb{E}[T_{k-1}Y_k]  - 2 c_0 \mathbb{E}[Y_k] \leq 0, \\
                        \hspace{-3mm}&&\hspace{-2mm}   \text{Equations}~\eqref{prob:1_1}~\text{to}~\eqref{prob:1_3}.
\end{eqnarray*}

\begin{algorithm}[!t]
\caption{Bisection Solution to Problem \textbf{P}$_2$}
\begin{algorithmic}[1]\label{Alg:bisec}
\REQUIRE ~~\\
    \STATE Set $l=0$, $u$ to be sufficiently large, and $\epsilon$ to be reasonably small;

\ENSURE ~~\\

\WHILE{$u-l>\epsilon$}
\STATE $c_0=(l+u)/2$;
\STATE Solve Problem \textbf{P}$_4$ using solver \texttt{fminunc} and proper $\boldsymbol{\kappa}_0$;
 \IF {$\text{min}_{\boldsymbol{\kappa}} \widetilde{\varphi}(\boldsymbol{\kappa}) <0$}
 \STATE $u=c_0$;
 \ELSE
 \STATE $l=c_0$;
 \ENDIF
\ENDWHILE
 \STATE \textbf{Output:} $\boldsymbol{\kappa}^*$, $c_0^*$
\end{algorithmic}

\end{algorithm}

For a given $c_0$, if Problem $\textbf{P}_2$ is feasible, then we have $c_0\geq h_\text{opt}$.
    Conversely, if Problem $\textbf{P}_2$ is infeasible, we have $c_0< h_\text{opt}$. 
This motivates us to solve $f_\text{X}(x)$  by a two-layer algorithm as shown in \textit{Algorithm} \ref{Alg:bisec}.

    To be specific, for a given $c_0$ and a given type of distribution with parameter vector $\boldsymbol{\kappa}$ (e.g., $\boldsymbol{\kappa}=[\alpha,\sigma^2]^\text{T}$ for the folded-normal distribution), the inner layer checks the feasibility of Problem $\textbf{P}_2$ by solving $\boldsymbol{\kappa}$ from the following optimization problem:
\begin{align*}
 (\textbf{P}_3)~~~g(c_0) = \min\limits_{\boldsymbol{\kappa}}~~&
                    \mathbb{E}[Y_k^2] +2\mathbb{E}[T_{k-1}Y_k]  - 2c_0 \mathbb{E}[Y_k] \qquad\\
    \text{s.~t.} ~~& \rho<1.
\end{align*}
It is clear that Problem $\textbf{P}_2$ is feasible only if the solution to Problem $\textbf{P}_3$ satisfies $g(c_0)\leq0$.
    In particular, $g(c_0)=0$ implies $c_0=h_\text{opt}$.
This is because for any other $c_0'<c_0$, we have $g(c_0')>0$, i.e., Problem \textbf{P}$_2$ is infeasible.

In the outer layer, we update $c_0$ using the bisection method, where the initial value of $c_0$  is set to be no less than $h_\text{opt}$.
    In each iteration,  $c_0$ is decreased if $g(c_0)<0$ (i.e., $\textbf{P}_2$ is feasible) and is increased otherwise (i.e., $\textbf{P}_2$ is infeasible).
Thus, $g(c_0)$ will converge to zero and the corresponding solution $\boldsymbol\kappa^*$ to Problem $\textbf{P}_3$ specifies the optimal $f_\text{X}(x)$ of the given type.

In most cases, the objective function $\varphi(\boldsymbol{\kappa})=\mathbb{E}[Y_k^2] +2\mathbb{E}[T_{k-1}Y_k]  - 2 c_0 \mathbb{E}[Y_k] $ of Problem \textbf{P}$_3$ contains transcendental functions and implicit expressions (e.g., $\rho_1$), and thus is mathematically intractable.
Therefore, we solve the problem using Matlab tools, e.g., the unconstrained solver \texttt{fminsearch}.
    In order to transform Problem \textbf{P}$_3$ into an unconstrained optimization problem, we consider the following objective function instead,
    \begin{equation*}
        (\textbf{P}_4) \qquad \widetilde{\varphi}(\boldsymbol{\kappa}) = \varphi(\boldsymbol{\kappa}) + MI_{\rho\geq1}, \qquad\qquad\qquad\qquad
    \end{equation*}
    where $M$ is a very large number and $I_\mathcal{A}$ is the indicator function, i.e., $I_\mathcal{A}$ equals to one if $\mathcal{A}$ is true and zero otherwise.
It is clear that Problem \textbf{P}$_4$ and Problem \textbf{P}$_3$ share the same solution.

\begin{example}
    If  $f_\text{X}(x)$ is a uniform distribution with $\boldsymbol{\kappa}_\text{u}=[\beta]$, we have
    \begin{equation*}
        f_\text{X}(x) = \frac{1}{\beta},\quad x\in(0,\beta).
    \end{equation*}

     In this case, the probability for $X_k$ being very large is zero and we have
    \begin{eqnarray*}
        \rho \hspace{-2.75mm}&=& \hspace{-2.75mm}\frac{2}{\beta\mu},~~\mathbb{E}[X_k] = \frac{\beta}{2},~~\mathbb{E}[X_k^2] =\frac{\beta^2}{3}, \\
        q_1\hspace{-2.75mm}&=&\hspace{-2.75mm}\frac{1}{\beta\mu^2(1-\rho_1)^2}-\frac{e^{-\mu(1-\rho_1)\beta}}{\mu(1-\rho_1)}\Big(1+\frac{1}{\beta\mu(1-\rho_1)}\Big).
    \end{eqnarray*}
Note that $\rho_1$ can be obtained efficiently by \textit{Algorithm} \ref{alg:rho_a}.
        Then we can express $\mathbb{E}[Y_k]$, $\mathbb{E}[Y_k^2]$,  and $\mathbb{E}[T_{k-1}Y_k]$ explicitly.
  Finally,we can solve the optimal $\beta^*$ using \textit{Algorithm} \ref{Alg:bisec}.
    $\hfill{} \blacksquare$
\end{example}

\begin{example}
    If  $f_\text{X}(x)$ is a Lomax distribution with $\boldsymbol{\kappa}_\text{l}=[\alpha, \beta]$, we have
    \begin{equation*}
        f_\text{X}(x) = \frac{\alpha\beta^\alpha}{(x+\beta)^{\alpha+1}} ,\quad x\geq0,
    \end{equation*}
    which is a heavy-tail distribution with probability $\Pr\{X_k>x\}$ decreasing polynomially.
     We have
    \begin{eqnarray*}
        \mathbb{E}[X_k] \hspace{-2.75mm}&=&  \hspace{-2.75mm} \frac{\beta}{\alpha-1}~\text{and}~\rho = \frac{\alpha-1}{\beta\mu}, \\
        \mathbb{E}[X_k^2] \hspace{-2.75mm}&=&  \hspace{-2.75mm}  \frac{2\beta^2}{(\alpha-1)(\alpha-2)} , 
    \end{eqnarray*}
    which is well defined for $\alpha>2$.
        In this case, although $q_1$ cannot be expressed explicitly, we can obtain $q_1$ numerically and then implement \textit{Algorithm} \ref{Alg:bisec} readily.
    $\hfill{} \blacksquare$
\end{example}

\begin{example}
    Suppose $f_\text{X}(x)$ is a folded-normal distribution with $\boldsymbol{\kappa}_\text{n}=[\alpha, \sigma]$,
    \begin{equation*}
        f_\text{X}(x) = \frac{1}{\sqrt{2\pi}\sigma} \Big( e^{\frac{-(x-\alpha)^2}{2\sigma^2}} + e^{\frac{-(x+\alpha)^2}{2\sigma^2}} \Big),\quad x\geq0.
    \end{equation*}
        It is clear that the tailing probability $\Pr\{X_k>x\}$ approximately decreases with $x$ at a speed of $\exp({x^2})$.
    We have
    \begin{eqnarray*}
        \mathbb{E}[X_k] \hspace{-2.75mm}&=&  \hspace{-2.75mm} \frac{2\sigma}{\sqrt{2\pi}} e^{\frac{-\alpha^2}{2\sigma^2}} + \alpha\Big( 1-2\Phi\Big(\frac{-\alpha}{\sigma}\Big) \Big), \\
        \mathbb{E}[X_k^2] \hspace{-2.75mm}&=& \hspace{-2.75mm} \alpha^2+\sigma^2 , \\
        \rho \hspace{-2.75mm}&=&  \hspace{-2.75mm} \frac{1}{\mu \mathbb{E}[X_k]},
    \end{eqnarray*}
where $\Phi(x)=\frac{1}{\sqrt{2\pi}} \int_{-\infty}^x e^{-\frac{t^2}{2}} \text{d}t$.

The MGF $G_\text{X}(t)=\mathbb{E}[e^{tX}]$ of $f_\text{X}(x)$ is given by
\begin{eqnarray*}
    G_\text{X}(t) \hspace{-2.75mm}&=&  \hspace{-2.75mm}  e^{\frac{\sigma^2t^2}{2}+\alpha t} \Big( 1-\Phi\Big( -\frac{\alpha}{\sigma}-\sigma t \Big) \Big) \\
            \hspace{-2.75mm}&&  \hspace{-2.75mm} + e^{\frac{\sigma^2t^2}{2}-\alpha t} \Big( 1-\Phi\Big( \frac{\alpha}{\sigma}-\sigma t \Big) \Big).
\end{eqnarray*}

    We then have
\begin{equation*}
    \rho_1 = G_\text{X}(-\mu(1-\rho_1)) ~\text{and}~ q_1 = G'_\text{X}(-\mu(1-\rho_1)).
\end{equation*}
By using Proposition \ref{prop:moments_Yk}, \textit{Algorithm} \ref{alg:rho_a}, and \textit{Algorithm} \ref{Alg:bisec}, we can calculate $\rho_1$, $\mathbb{E}[Y_k]$, $\mathbb{E}[Y_k^2]$, $\mathbb{E}[T_{k-1}Y_k]$, $q_1$, and solve optimal $\boldsymbol{\kappa_\text{n}^*}$ efficiently.
    $\hfill{} \blacksquare$
\end{example}

\begin{example}
    Suppose $f_\text{X}(x)$ is an exponential distribution with $\boldsymbol{\kappa}_\text{e}=[\lambda]$, i.e.,
    \begin{equation*}
        f_\text{X}(x) = \lambda e^{-\lambda x},\quad x\geq0,
    \end{equation*}
   where $\Pr\{X_k>x\}$ decreases exponentially with $x$.

    For a given service rate $\mu$, we have $\rho_1=\rho=\lambda/\mu$ and
    \begin{eqnarray}
        \nonumber
        \hspace{-6mm}&&\mathbb{E}[T_{k-1}Y_k]  = \frac{1}{\mu^2(1-\rho)} + \frac{1-\rho}{\mu^2\rho}, \\
        \label{rt:aoi_exp}
        \hspace{-6mm}&& \widebar{\Delta}_\text{D} =\frac1\mu\Big(1+\frac1\rho+\frac{\rho^2}{1-\rho}\Big),
    \end{eqnarray}
    where the average AuD $\widebar{\Delta}_\text{D}$ has the same expression as the average AoI $\widebar{\Delta}$ and the AuD-optimal arrival rate $\lambda^*$ is close to $\mu/2$ \cite{Yates-2012-age}.
    $\hfill{} \blacksquare$
\end{example}

\vspace{1mm}
\noindent~\textit{3) Missing Probability of Updates}
\vspace{1mm}

Although  increasing the decision rate does not reduce the average AuD, increasing the decision rate does help reduce the missing probability.
    To be specific, fewer updates would be missed to make decisions if the decision rate is increased.

\begin{definition}
        \textit{Missing probability} $p_\text{mis}$ of updates is the limiting ratio between the number of updates missed for  decisions and the number of total received updates as the length of  the considered period goes to infinity.
\end{definition}

Equivalently, missing probability $p_\text{mis}$ is the probability of no decision being made during each inter-departure time $Y_k$, i.e., $N_k=0$.
    Given that $Y_k=y$, we know that $N_k$ follows the Poisson distribution with parameter $\nu y$.
By taking the expectation over $Y_k$, the missing probability could be obtained readily, as shown in the following proposition.

\begin{proposition} \label{prop:missing_prob}
        For $G/M/1/M$ update-and-decide systems,  the missing probability of updates is given by,
        \begin{equation*}
                p_\text{mis} = \frac{\mu\big(\mu(1-\rho_1)q_0-\nu \rho_1\big)}  {(\mu+\nu)\big(\mu(1-\rho_1)-\nu\big)},
        \end{equation*}
where $\rho_1$ is given by \eqref{df:rho_a} and
\begin{equation*}
    q_0 = \int_0^\infty f_\text{X} (x)  e^{-\nu x} \text{d} x.
\end{equation*}
\end{proposition}

\begin{proof}
    See Appendix \ref{apx:missing_prob}.
\end{proof}

While the average AuD quantifies the timeliness of updates,  the missing probability specifies the utilization (or efficiency) of the received updates.
    Thus, the performance of an update-and-decide system can be well characterized by average AuD and missing probability.
In particular, the missing probability of an $M/M/1/M$ update-and-decide system is explicitly given by
\begin{equation*}
    p_\text{mis} = \frac{\lambda}{\lambda+\nu}.
\end{equation*}

\begin{figure}[htp]   

\hspace{-6 mm}
    \begin{tabular}{cc}
    \subfigure[Minimum average AuD versus service rate $\mu$]
    {
    \begin{minipage}[t]{0.5\textwidth}
    \centering
    {\includegraphics[width = 3.5in] {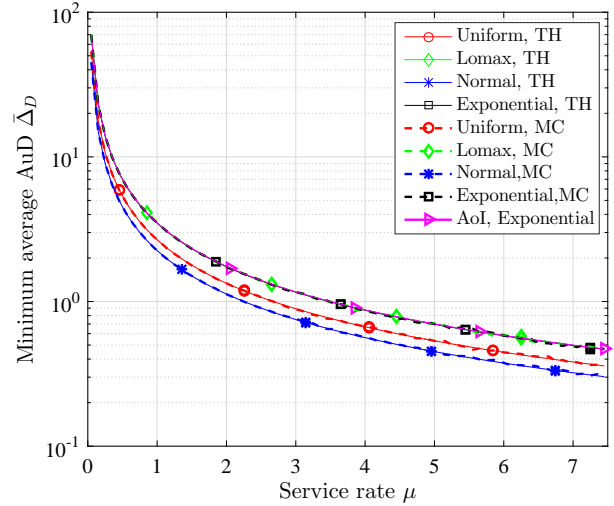} \label{fig:random_aud_mu}}
    \end{minipage}
    }\\

    \subfigure[Achievable arrival rate  versus service rate $\mu$]
    {
    \begin{minipage}[t]{0.5\textwidth}
    \centering
    {\includegraphics[width = 3.5in] {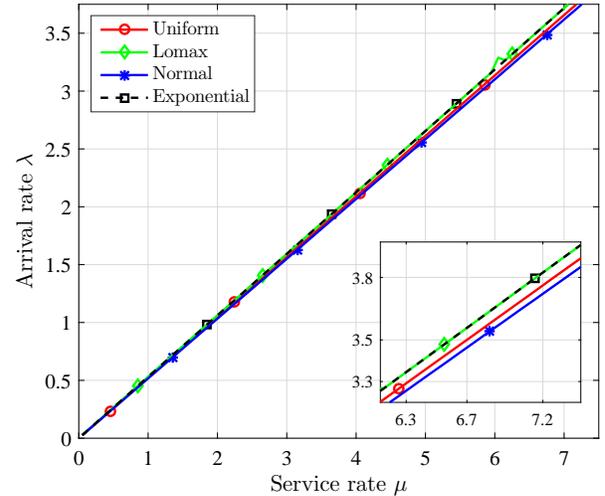} \label{fig:random_lambda_mu}}
    \end{minipage}
    }
    \end{tabular}
\caption{Performance of the optimal arrivals in a $G/M/1/M$ update-and-decide system, where `TH' represents a theocratical result and `MC' represents a Monte Carlo result. } \label{fig:random_arrival}
\end{figure}

\subsection{Numerical Results} \label{subsec:numerical_1}

\begin{figure*}[htp]   

\hspace{-6 mm}
    \begin{tabular}{cc}
    \subfigure[Missing probability versus decision rate.]
    {
    \begin{minipage}[t]{0.5\textwidth}
    \centering
    {\includegraphics[width = 3.5in] {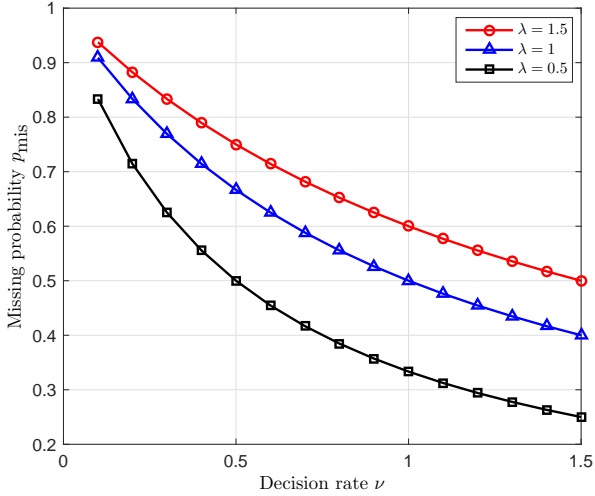} \label{fig:pmis_vs_nu}}
    \end{minipage}
    }

    \subfigure[Average AuD versus arrival rate.]
    {
    \begin{minipage}[t]{0.5\textwidth}
    \centering
    {\includegraphics[width = 3.5in] {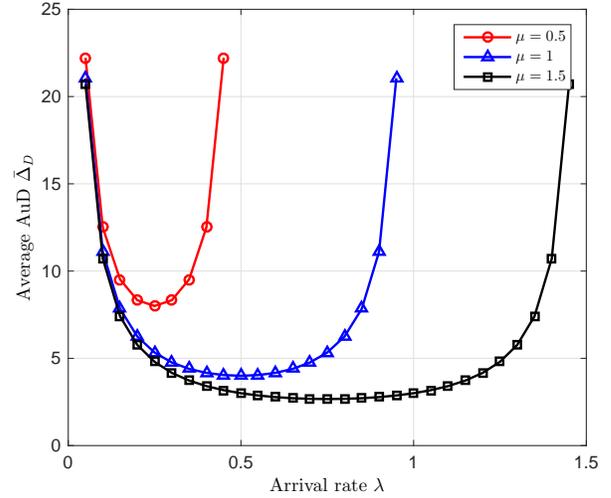} \label{fig:delta_vs_lambda}}
    \end{minipage}
    }
    \end{tabular}
\caption{Utilization and timeliness of an $M/M/1/M$ update-and-decide systems. }
\end{figure*}

We consider the following four types of arrival processes: the uniform arrival process, the Lomax arrival process,  the folded-normal arrival process, and the exponential arrival process.
    For each arrival process and each service rate $\mu$, we  search the optimal distribution parameter $\boldsymbol\kappa^*$  using \textit{Algorithm} \ref{Alg:bisec}.
        Using the obtained $\boldsymbol\kappa^*$, we then present how the minimum average AuD changes with service rate $\mu$ in Fig. \ref{fig:random_aud_mu}.

\textit{First}, we observe that the average AuD of a system employing a folded-normal arrival process with $\boldsymbol\kappa_\text{n}=[\alpha,\sigma^2]$ is the smallest.
    This is because for the folded-normal distribution,  the probability $\Pr\{X_k>x\}$ decreases very fast, especially when $\sigma^2$ is small.
Actually, our results show that for each $\mu$,  the optimal variance is $\sigma^2=0$  (with fluctuations less than $10^{-5}$), which means that $X_k$ degrades to a constant equal to $\alpha$.
    In fact, the periodicity of arrivals eliminates all their uncertainty and thus is beneficial in reducing average AuD.
    Therefore, a well scheduled (with a properly chosen rate) periodic arrival process is a simple yet AuD-optimal choice.
\textit{Second}, the average AuDs are almost the same for a system with a Lomax arrival process and a system with an exponential arrival process, and are slightly larger than that of a system with a folded-normal arrival process.
    On the one hand, the Lomax distribution has two parameters and the probability $\Pr\{X_k>x\}$ can be tuned to decrease (with $x$) as quickly as that of an exponential distribution.
On the other hand, neither the Loamx distribution nor the exponential distribution can be tuned to be close to a deterministic constant like the folded-normal distribution.
    \textit{Third}, the AuD of a system with a uniform arrival process is larger than a system with a folded-normal arrival process, but is smaller than a system with a Lomax or an exponential arrival process.
This is because for the uniform arrival process, inter-arrival time $X_k$ can never be larger than $\beta$.
    Seeing from  $\mathbb{R}^+$, the uncertainty of $X_k$ is relatively low, which is beneficial in optimizing the average AuD.
\textit{Fourth}, we also present the average AoI  of an $M/M/1$ queue with optimal arrival rate by the triangle labeled pink curve, which exactly coincide with the average AuD of an update-and-decide system with exponential arrivals and Poisson decisions (see \eqref{rt:aoi_exp} and the square labeled black curve).
    \textit{Fifth}, it is seen that our theoretical results exactly coincide with Monte Carlo simulation results, which further validates our analysis.

Fig. \ref{fig:random_lambda_mu} depicts the maximal achievable arrival rate $\lambda=1/\mathbb{E}(X_k)$ of the considered systems when the average-AuD-minimizing  $\boldsymbol\kappa^*$ is  used.
    It is observed that the arrival rate increases approximately linearly with the service rate.
Roughly speaking, we have $\lambda\thickapprox\mu/2$ for all the update-and-decide systems under consideration.
    If the service rate is increased, therefore, the DMU would receive updates at a higher rate, which means that the average AuD would be smaller.

Fig. \ref{fig:pmis_vs_nu} plots how missing probability $p_\text{mis}$ varies with decision rate $\nu$ in an $M/M/1/M$ update-and-decide system, where the arrival rate is $\mu=2$.
    It is seen that as $\nu$ is increased, $p_\text{mis}$ decreases quickly.
It is also observed that if we set the arrival rate $\lambda$ to be smaller, $p_\text{mis}$ decreases more quickly.
    Fig. \ref{fig:delta_vs_lambda} presents how  average AuD changes with arrival rate $\lambda$ in an $M/M/1/M$ updating system.
As is shown, the average AuD is large when $\lambda$ is either very small or very large.
    To be specific, when $\lambda$ is small, the average AuD is large because the waiting time for the arrival of a new update is large.
When $\lambda$ is large, the queueing delay of updates is large due to the limited service capability of the server.
To minimize AuD, therefore, we should try to increase the service rate and set the arrival rate to be close to (while smaller than) a half of the service rate, i.e., $\lambda\thickapprox\mu/2$.

\section{Average AuD with Periodic Arrivals}\label{sec:4_aud_determine}
Although we have observed that the updating system with folded-normally distributed inter-arrival time achieves the minimum average AuD as  variance $\sigma^2$ approaches zero, it is not clear whether the deterministic process is also optimal for decisions.
    In this section, we shall investigate the average AuD and the optimal decision scheduling for updating systems with periodic arrivals.

Consider a $D/M/1/G$ update-and-decide system with deterministic inter-arrival time and exponentially distributed service time, i.e., $f_\text{X}(x)=\delta(x-\frac1\lambda)$ and $f_\text{S}(x)=\mu e^{-\mu x}$.
    Thus, we have $X_k=1/\lambda$ with probability 1.

With  periodic arrivals, we denote
\begin{eqnarray}
\nonumber
    \rho_0 \hspace{-2.75mm} &=&\hspace{-2.75mm} \int_0^\infty f_\text{X}(x) e^{-\mu x} \text{d}x= e^{\frac{-\mu}{\lambda}}, \\
    \label{rt:rho_1_peri}
    \rho_1  \hspace{-2.75mm} &=&\hspace{-2.75mm} \int_0^\infty f_\text{X}(x) e^{-\mu(1-\rho_1)x} \text{d}x = e^{\frac{-\mu(1-\rho_1)}{\lambda}}.
\end{eqnarray}

By solving $\rho_1$ from \eqref{rt:rho_1_peri}, we further have
\begin{equation*}
    \rho_1= -\rho W_0\Big(\frac{-1}{\rho}e^{\frac{-1}{\rho}}\Big),
\end{equation*}
    where $\rho=\lambda/\mu$ and $W_0(\cdot)$ is the 0-th branch of the Lambert function.

Some frequently used notations of this section are summarized in the following table.
\vspace{2mm}
\begin{table}[h]
\caption{Notations}
\vspace{-1.0mm}
\centering
\begin{tabular}{l|l}
\toprule[1pt]
                   $\rho_1=e^{\frac{-\mu(1-\rho_1)}{\lambda}}$          & $\rho_0=e^{\frac{-\mu}{\lambda}}$ \\
                   $w_1=e^{\frac{-\mu(1-\rho_1)}{\nu}}$          & $w_0=e^{\frac{-\mu}{\nu}}$ \\
                   $u_1=e^{-\mu(1-\rho_1)\delta}$          & $u_0=e^{-\mu\delta}$ \\
                   $~~\rho=\frac\lambda\mu$          &  \\
\bottomrule[1pt]
\end{tabular}
\label{tab:notations}
\end{table}

\subsection{Average AuD of D/M/1/M Update-and-Decide Systems}
In this subsection, we assume that decisions are made following a Poisson process.
    Thus, the inter-decision times are exponentially distributed with pdf  $f_\text{Z}(x)=\nu e^{-\nu x}$.

    Following \textit{Theorem} \ref{thm:thm_1_aud_g_1} and \textit{Proposition} \ref{prop:moments_Yk}, we have the following immediate corollary.

\begin{corollary} \label{coro:dm1_pois}
    In a $D/M/1/M$ update-and-decide system with arrival rate $\lambda$, service rate $\mu$, and Poisson decisions of rate $\nu$, the average AuD is given by
    \begin{equation*}
        \widebar\Delta_\text{D} =\frac{1}{2\lambda}+ \frac{1}{\mu(1-\rho_1)}.
    \end{equation*}
\end{corollary}

\begin{proof}
    By combing the results in \textit{Theorem} \ref{thm:thm_1_aud_g_1} and \textit{Proposition} \ref{prop:moments_Yk}, the corollary can readily be proved.
\end{proof}

    Although the optimal $\lambda^*$ cannot be presented explicitly, it can be obtained using \textit{Algorithm} \ref{alg:rho_a} and is expected to close to $0.5$.
When $\lambda^*$ is used, the corresponding average AuD  would coincide with the minimum average AuD  of a $G/M/1/M$ system with a folded-normal arrival process.

\subsection{Average AuD of D/M/1/D Update-and-Decide Systems}
In this subsection, we study the average AuD and the optimal scheduling of a $D/M/1/D$ update-and-decide system in which the decision process is \textit{periodic}.
    In such systems, the decision epochs are no longer uniformly distributed within each inter-departure time $Y_k$.
Thus, the AuD performance of the system needs to be considered separately.

We denote the inter-arrival time as $X_k=1/\lambda$ and the inter-decision time as $Z_j=1/\nu$.
    In particular, we assume that $\nu$ is an integer multiple of $\lambda$, i.e., $\nu=m_0\lambda$, where $m_0$ is a positive integer.
That is, $m_0$ decisions are made during each inter-arrival time.
    To keep the missing probability $p_\text{mis}$ low, we set $m_0$ to be no less than one, i.e., $m_0\geq1$.

\vspace{1mm}
\noindent~\textit{1) Average AuD of Synchronous Systems}
\vspace{1mm}

First, we consider the case where the arrival process and the decision process are \textit{synchronous}, i.e., each update arrival epoch is aligned with a certain decision epoch.
    The average AuD of the system is presented in the following theorem.

\begin{theorem}\label{thm:aud_dm1_periodic}
In a synchronous $D/M/1/D$ update-and-decide system with arrival rate $\lambda$, service rate $\mu$, and periodic decisions with rate $\nu=m_0\lambda$, the average AuD is given by
    \begin{equation} \label{rt:p_p_aud}
        \widebar{\Delta}_\text{D}  =\frac{1+m_0}{2\nu} + \frac{w_1}{\nu(1-w_1)},
    \end{equation}
    where $w_1=e^{-\frac{\mu(1-\rho_1)}{\nu}}$.
\end{theorem}

\begin{proof}
    See Appendix \ref{apx:aud_dm1_periodic}.
\end{proof}

It can be readily verified that \eqref{rt:p_p_aud} can be rewritten as
\begin{equation*}
    \widebar{\Delta}_\text{D}  =\frac{1}{2\lambda} + \frac{1}{m_0\lambda} \left( \frac12 + \frac{1}{e^{\frac{\mu(1-\rho_1)}{m_0\lambda}}-1}\right).
\end{equation*}

    Since the last item (one over $(\exp(\mu(1-\rho_1)/m_0\lambda)-1)$) is increasing with $m_0$ more slowly than linear, $\widebar{\Delta}_\text{D}$ would be decreasing with $m_0$.
As $m_0$ (the decision rate) goes to infinity, $\widebar{\Delta}_\text{D}$ would eventually decrease to the average AoI $\widebar{\Delta}$ of the corresponding $D/M/1$ updating system.
    Note that $\widebar{\Delta}$ is equal to the average AuD of a $D/M/1/M$ update-and-decide system with Poisson decisions (see the remarks after \textit{Theorem} \ref{thm:thm_1_aud_g_1}).
Thus, it is concluded that the average AuD of a $D/M/1/M$ system where Poisson decisions are made is smaller than that of the corresponding $D/M/1/D$ system where periodic decisions are made.
    The reason is that the periodicity of decisions would greatly limit the flexibility of decisions to explore the potential freshness of received updates.

We also have the following results on the missing probability of updates.
\begin{proposition}\label{prop:missing_prob_periodic}
    In a synchronous $D/M/1/D$ update-and-decide system with arrival rate $\lambda$, service rate $\mu$, and periodic decisions with rate $\nu=m_0\lambda$, the missing probability of updates is given by
    \begin{equation*}
        p_\text{mis} = \frac{\rho_1}{2-\rho_1}\Big(\frac{1}{w_1}-w_0\Big).
    \end{equation*}
\end{proposition}

\begin{proof}
    See Appendix \ref{apx:missing_prob_periodic}.
\end{proof}

\vspace{1mm}
\noindent~\textit{2) Optimal Decision Scheduling in Asynchronous Systems}
\vspace{1mm}

    Second, although the periodicity of the decision process slightly increases average AuD, we shall show that by scheduling the decisions properly, the periodic decision process can then achieve a smaller average AuD than the Poisson decision process does.
In particular, we consider a $D/M/1/D$ update-and-decide system where the arrival process and the decision process are asynchronous.
    In this subsection, we assume that the arrival rate equals the decision rate, i.e., $\nu=\lambda$.
That is, each update arrival epoch is followed by a decision epoch which is exactly $\delta$ seconds after it.
    We shall present the average AuD of the system and investigate how the offset $\delta$ affects it.

\begin{algorithm}[!t]
\caption{Iterative Solution to Offset $\delta$}
\begin{algorithmic}[1]\label{Alg:delta}
\REQUIRE ~~\\
    \STATE Set initial offset as $\delta=\frac{1}{2\lambda}$, initial $u_1$ as $u_{1}^{(0)}=e^{\frac{-(1-\rho_1)}{2\rho}}$, and $l=0$;
    \STATE Set initial step-size as $\varsigma=\frac{u_{1}}{4}$;
    \STATE Set $\epsilon$ to be reasonably small;
    \STATE  $\phi^{(0)}=\phi\big(u_1^{(0)}\big)$;

\ENSURE ~~\\

\WHILE {$\big| \phi^{(l)} \big|>\epsilon$ }
\STATE  $s = \text{sign} \big(\phi^{(l)}\big)$;
\STATE  $u_1^{(l+1)}=u_1^{(l)} + s\varsigma$;
\STATE  $l=l+1$;
\STATE  $\phi^{(l)}=\phi\big(u_1^{(l)}\big)$;

\IF {$s \phi^{(l)}>0$}
 \STATE $\varsigma=2\varsigma$;
 \ELSE
 \STATE $\varsigma=\frac\varsigma7$;
 \ENDIF

\ENDWHILE

 \STATE \textbf{Output:} $u_1^*=u_1^{(l)}$ and $\delta^*=\frac{-\ln(u_1^*)}{\mu(1-\rho_1)}$
\end{algorithmic}

\end{algorithm}

The average AuD of this system is presented in the following theorem.
\begin{theorem}\label{thm:aud_dm1_periodic_offset}
In an asynchronous $D/M/1/D$ updating system with periodic arrivals at rate $\lambda$, exponential services at rate $\mu$, periodic decisions at rate $\nu=\lambda$, and offset $\delta\in(0,1/\lambda)$, the average AuD is given by
    \begin{equation} \label{rt:thm_3_aud_off}
        \widebar{\Delta}_\text{D}=\delta + \frac{(1-\rho_0)u_1^2+(1-\rho_1)(1-u_0)u_1}{\lambda(1-\rho_1)(1-\rho_0)},
    \end{equation}
    where $u_0=e^{-\mu\delta}$, $u_1=e^{-\mu(1-\rho_1)\delta}$, and $\rho_0=e^{-\frac\mu\lambda}$.
\end{theorem}

\begin{proof}
    See Appendix \ref{apx:aud_dm1_periodic_offset}.
\end{proof}

The derivative of $\widebar{\Delta}_\text{D}$ over offset $\delta$ can be  presented as a function of $u_1$ as follows
\begin{equation*}
    \phi(u_1) = 1-\frac{2u_1^2}{\rho} - \frac{(1-\rho_1)u_1}{\rho(1-\rho_0)} +\frac{2-\rho_1}{\rho(1-\rho_0)} u_1^{\frac{2-\rho_1}{1-\rho_1}}.
\end{equation*}

Note that the average AuD is minimized when the arrival rate is close to $\lambda=\mu/2$.
    In this case, we have $\rho=0.5$, $\rho_0=e^{\frac{-1}{\rho}}=0.1353$, and $\rho_1=-\rho W_0(\frac{-1}{\rho} e^{\frac{-1}{\rho}})=0.2032$.
It can readily be verified that $\phi(0)<0$, $\phi(\frac1\lambda)>0$, and $ \phi(\delta)$ is convex.
    By increasing $\delta$ properly, therefore, the average AuD can  be minimized.
In particular, we can search the optimal $\delta$ efficiently using the iterative process shown in \textit{Algorithm} \ref{Alg:delta}.

From \textit{Theorem} \ref{thm:aud_dm1_periodic} and \textit{Theorem} \ref{thm:aud_dm1_periodic_offset}, we see that in $D/M/1/D$ updating systems, the average AuD is decreasing with decision rate and can be further minimized by optimizing the offset $\delta$.
    This would be useful for many engineering implementations.
In wireless sensor networks, for example, we can arrange periodic sensing tasks for sensors and collect the sensed information using a mobile agent (e.g., a UAV).
    With random delay (since the UAV may sometimes be out of the transmit range), the collected information is relayed to the monitor.
In this situation, the monitor should make periodic data processing/predictions with optimized offset so that the processing/prediction can be performed timely.

\begin{figure}[!t]

\hspace{-6 mm}
    \begin{tabular}{cc}
   \subfigure[Average AuD versus decision rate.]
    {
    \begin{minipage}[t]{0.5\textwidth}
    \centering
    {\includegraphics[width = 3.5in] {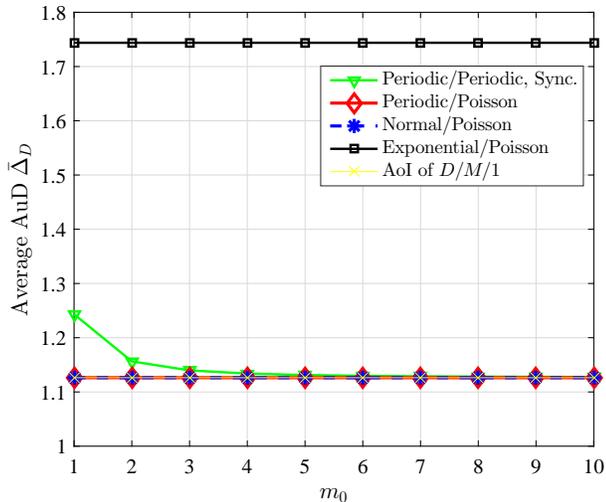} \label{fig:aud_peri_peri_m0}}
    \end{minipage}
    }\\

     \subfigure[Missing prob. versus decision rate.]
    {
    \begin{minipage}[t]{0.5\textwidth}
    \centering
    {\includegraphics[width = 3.5in] {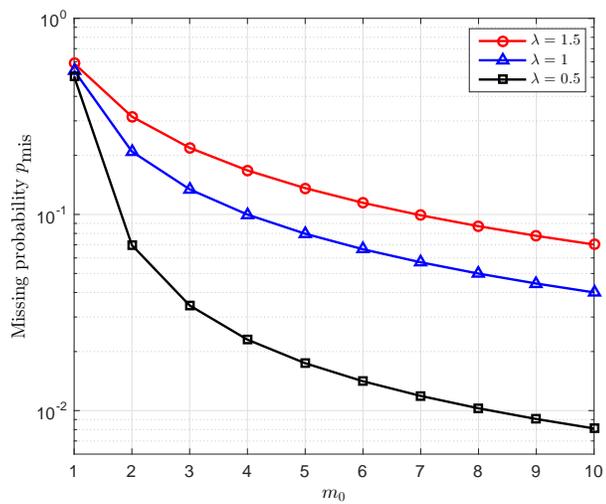} \label{fig:pmis_peri_peri_m0}}
    \end{minipage}
    }

    \end{tabular}

\caption{Performance in a synchronous $D/M/1/D$ update-and-decide system with periodic arrivals and periodic decisions. } \label{fig:aud_1_1}
\end{figure}

\subsection{Numerical Results}

We investigate how the average AuD of a synchronous $D/M/1/D$ updating system varies with decision rate $\nu=m_0\lambda$ (see \textit{Theorem} \ref{thm:aud_dm1_periodic}) in Fig. \ref{fig:aud_peri_peri_m0}.
    The service rate is set to $\mu=2$.
The arrival rate is set to $\lambda=0.5175\mu$, which minimizes the average AoI of the corresponding $D/M/1$ updating system~\cite{Yates-2012-age}.
    Our simulation results also show that this is also the AuD-optimal arrival rate for $D/M/1/M$ update-and-decide system.
It is observed in Fig. \ref{fig:aud_peri_peri_m0} that as $m_0$ is increased, the average AuD of the $D/M/1/D$ system (the triangle labeled curve) decreases slowly.
    As $m_0$ is increased to be relatively large (e.g., $m_0=5$), the average AuD decreases to that (the diamond labeled curve) of a $D/M/1/M$ system with periodic (or folded-normal) arrivals and Poisson decisions, which is equal to the average AoI (the yellow curve labeled by 'x') of corresponding $D/M/1$ queues.
    Moreover, the missing probability of updates is also decreasing with $m_0$, as shown in Fig. \ref{fig:pmis_peri_peri_m0}.
Different from $G/M/1/M$ update-and-decide systems where average AuD is independent of decision rate, therefore, using a reasonably large decision rate is preferred in $D/M/1/D$ update-and-decide systems.

In Fig. \ref{fig:aud_peri_peri_delta}, we study the average AuD of a $D/M/1/D$ update-and-decide system with asynchronous and periodic decisions (see \textit{Theorem} \ref{thm:aud_dm1_periodic_offset}), where the decision rate is set to be equal to the arrival rate, i.e., $\nu=\lambda$.
    It is observed that as the offset changes within $\delta\in(0,1/\lambda)$, the average AuD is firstly decreasing and then  increasing, i.e., is convex with $\delta$.
Using \textit{Algorithm} \ref{Alg:delta}, the optimal offset $\delta^*$ can be obtained efficiently, as shown by the cyan pentagram.
    Since neither $u_0$ nor $u_1$ is linear with $\delta$ (see \eqref{rt:thm_3_aud_off}), $\delta^*$ is close to but not equal to $1/2\lambda$, as is shown in Fig. \ref{fig:aud_peri_peri_delta}.
By using the optimal offset, the average AuD of a $D/M/1/D$ updating system with asynchronous decisions is actually smaller than the average AuD of all other updating systems, as well as the average AoI (the yellow curve labeled by 'x') of a corresponding $D/M/1$ queue.

\begin{figure}[!t]
\centering
\includegraphics[width=3.5in]{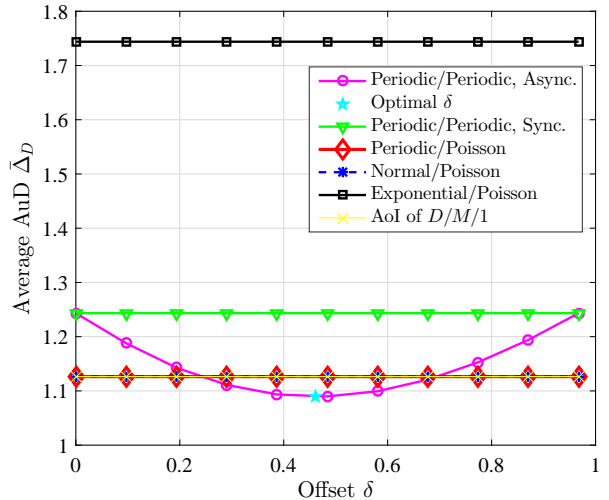}
\caption{Average AuD versus offset in an asynchronous $D/M/1/D$ systems with periodic arrivals and periodic decisions. } \label{fig:aud_peri_peri_delta}
\end{figure}

\section{Conclusion}\label{sec:5_conclusion}
In this paper, we proposed a new metric termed as age upon decisions to evaluate the freshness of received updates at decision epochs.
    We showed that for $G/M/1/M$ update-and-decide systems in which decisions are made randomly, the average AuD is minimized by the periodic (i.e., deterministic) arrival process, which completely eliminates the uncertainty of arrivals.
For $D/M/1/D$ update-and-decide systems in which decisions are made periodically (i.e., deterministically), we show that the average AuD of the system is larger than that of corresponding $D/M/1/M$ systems, since the randomness of the decision process can provide some flexibility for the monitor.
    By increasing the decision rate of a $D/M/1/D$ system, however, the average AuD can be reduced, which finally converges to that of a corresponding $D/M/1/M$ system.
Moreover, the average AuD can be further decreased by properly controlling the offset between the arrival process and the decision process.
    In a nutshell, it is suggested to employ periodic arrival processes and Poisson decision processes, unless the decision rate can be sufficiently large or the arrival-decision offset can be properly controlled.

\appendix

\renewcommand{\theequation}{\thesection.\arabic{equation}}
\newcounter{mytempthcnt}
\setcounter{mytempthcnt}{\value{theorem}}
\setcounter{theorem}{2}
\addcontentsline{toc}{section}{Appendices}\markboth{APPENDICES}{}

\subsection{Proof of Theorem \ref{thm:thm_1_aud_g_1}}  \label{apx:prf_1_aud_g_1}
\begin{proof}
    Given an inter-departure time $Y_k=y$, suppose $N_k$ decisions are made at epochs $\{\tau_{k_j}, j=1,2,\cdots,N_k\}$.
        It is clear that $N_k$ is a Poisson distributed random number with parameter $\nu y$.
    That is, the probability that $n$ decisions are made during $Y_k$ is
\begin{equation*}
    \Pr\{N_k=n|Y_k=y\}=\frac{(\nu y)^n}{n!}e^{-\nu y}.
\end{equation*}

We denote $\tau'_{k_j}=\tau_{k_j}-t'_{k-1}$.
    Since decision epochs $\tau_{k_j}$ are independently and uniformly distributed in $Y_k$,  $\tau'_{k_j}$ would be independently and uniformly distributed over $[0,y]$.
Thus, the expected sum $\Delta'_{\text{D}k}=\sum_{j=1}^n \tau'_{k_j}$   can be  expressed as
\begin{equation*}
    \mathbb{E}\big[\Delta'_{\text{D}k}|Y_k=y,N_k=n\big] =\sum_{j=1}^n \mathbb{E}[\tau'_{k_j}]=\frac{ny}{2}.
\end{equation*}

Since the AuD at decision epoch $\tau_{k_j}$ is $\Delta_\text{D}(\tau_{k_j})=T_{k-1}+\tau'_{k_j}$, where $T_{k-1}$ is the system time of the latest received update, the expectation of sum AuD $\Delta_{\text{D}k}=\sum_{j=1}^{N_k} \Delta_\text{D}(\tau_{k_j})$ would be
\begin{eqnarray*}
    \mathbb{E}[\Delta_{\text{D}k}|Y_k=y]
    \hspace{-2.75mm}&=&\hspace{-2.75mm}  \sum_{n=0}^\infty \Pr\{N_k=n|Y_k=y\} \left(\frac{ny}{2}  + n\mathbb{E}[T_{k-1}]\right) \\
    \hspace{-2.75mm}&=&\hspace{-2.75mm}  \frac\nu2  (y^2 + 2y\mathbb{E}[T_{k-1}]).
\end{eqnarray*}

    Taking the expectation over $Y_k$, we have
\begin{eqnarray*}
    \mathbb{E}[\Delta_{\text{D}k}] =\frac{\nu}{2}\mathbb{E}[Y_k^2] + \mathbb{E}[T_{k-1}Y_k].
\end{eqnarray*}

Assume that there are $K$  departure epochs and $N_\text{T}$ decision epochs during a period $T$, we have $N_\text{T}=\sum_{k=1}^K N_k$.
    As $T$ goes to infinity, we have
\begin{eqnarray*}
        \nonumber \widebar{\Delta}_\text{D} \hspace{-2.75mm}&=&\hspace{-2.75mm}
                                \lim_{T\rightarrow\infty} \frac{1}{N_\text{T}}\sum_{k=1}^K \Delta_{\text{D}k}
                                = \lim_{T\rightarrow\infty} \frac {K}{N_\text{T}} \frac1K\sum_{k=1}^K \Delta_{\text{D}k}
                                = \frac{\mathbb{E}[\Delta_{\text{D}k}] }{\nu\mathbb{E}[Y_k]} \\
       \label{apx:rt_aud}  \hspace{-2.75mm}&=&\hspace{-2.75mm} \frac{\mathbb{E}[Y_k^2] +2\mathbb{E}[T_{k-1}Y_k]}{2\mathbb{E}[Y_k]}.
\end{eqnarray*}

Thus, the proof of \textit{Theorem} \ref{thm:thm_1_aud_g_1} is completed.
\end{proof}

\subsection{Proof of Proposition \ref{prop:missing_prob}} \label{apx:missing_prob}
\begin{proof}
The probability that no decision is made during $Y_k=y$ is $\Pr\{N_k=0\}=e^{-\nu y}$.
    By taking expectation over $Y_k$, the missing probability can be expressed as
    \begin{eqnarray*}
            \hspace{-2.75mm}&&\hspace{-2.75mm} p_\text{mis}=\mathbb{E}[e^{-\nu Y_k}]\\
                           \hspace{-2.75mm}&&\hspace{-2.75mm}   = \rho_1 \mathbb{E}[e^{-\nu S_k}]  + (1-\rho_1) \mathbb{E}[e^{-\nu (X_k-T_{k-1}+S_k)}|X_k> T_{k-1}] \\
                           \hspace{-2.75mm}&&\hspace{-2.75mm} =  \mathbb{E}[e^{-\nu S_k}] \left(\rho_1 + \int_0^\infty f_\text{X}(x) \text{d}x \int_0^x f_\text{T}(t) e^{-\nu(x-t)}\text{d}t \right) \\
                           \hspace{-2.75mm}&&\hspace{-2.75mm} =   \frac{\mu}{\mu+\nu} \cdot \frac{\mu(1-\rho_1)q_0-\nu\rho_1}{\mu(1-\rho_1)-\nu},
        \end{eqnarray*}
where $q_0 = \int_0^\infty f_\text{X} (x)  e^{-\nu x} \text{d} x$ and $\rho_1=\Pr\{X_k\leq T_{k-1}\}$  is defined in \eqref{df:rho_a}.
    This complete the proof of \textit{Proposition} \ref{prop:missing_prob}.
\end{proof}

\subsection{Proof of Theorem \ref{thm:aud_dm1_periodic}} \label{apx:aud_dm1_periodic}
\begin{proof}
In this proof, notations in Table \ref{tab:notations} are used.

    Firstly, consider the case where the system time is not less than the inter-arrival time, i.e., $X_k\leq T_{k-1}$.
As $X_k=1/\lambda$, we have
\begin{eqnarray*}
    \hspace{-6mm}&& \rho_1 = \Pr\{X_k\leq T_{k-1}\} = \int_{\frac1\lambda}^\infty f_\text{T} (x) \text{d}x =e^{\frac{-\mu(1-\rho_1)}{\lambda}}, \\
    \hspace{-6mm}&& \mathbb{E}[T_{k-1}|X_k\leq T_{k-1}] = \frac{1}{\rho_1} \int_{\frac1\lambda}^\infty x f_\text{T}(x)\text{d}x = \frac1\lambda+\frac{1}{\mu(1-\rho_1)},
\end{eqnarray*}
where $f_\text{T}(x)$ is given by \textit{Proposition} \ref{prop:fx_Tk}.

As shown in Fig. \ref{fig:periodic}(a), we have $Y_k=S_k$ and
    \begin{equation*}
        f_{\text{Y}|X\leq T}(y) = \mu e^{-\mu y}, \quad y>0.
    \end{equation*}

Suppose $T_{k-1}$ consists of $j$ decision intervals, i.e., $\frac j\nu\leq T_{k-1}<\frac{j+1}{\nu}$.
    We denote $r_j=T_{k-1}-\frac j\nu$ and $s_j=\frac1\nu-r_j$ (see Fig. \ref{fig:periodic}(a))\footnote{In Appendix E and Appendix F, we denote the residual part of an inter-decision time as $s_j$, which is slightly consuming with  the service time $S_k$.}.
For $j=m_0, m_0+1,\cdots$, we further denote $\varrho_j=\Pr\{\frac j\nu\leq T_{k-1}<\frac{j+1}{\nu}\}$ and have
\begin{equation*}
    \varrho_j = \int_{\frac j\nu}^{\frac{j+1}{\nu}} f_\text{T}(x) \text{d}x =  (1-w_1)w_1^j,
\end{equation*}
where $f_\text{T}(x)$ is given in \textit{Proposition} \ref{prop:fx_Tk} and $w_1=e^{\frac{-\mu(1-\rho_1)}{\nu}}$.

Under the condition $\frac j\nu\leq T_{k-1}<\frac{j+1}{\nu}$, we have
\begin{eqnarray*}
    \hspace{-2.75mm}&&\hspace{-2.75mm}\Pr\left\{s_j\leq x \bigg|\frac j\nu\leq T_{k-1}<\frac{j+1}{\nu} \right\} \\
     \hspace{-2.75mm}&&\hspace{-2.75mm}  = \Pr\bigg\{ T_{k-1}\geq \frac{j+1}{\nu}-x \bigg|\frac j\nu\leq T_{k-1}<\frac{j+1}{\nu} \bigg\} \\
                \hspace{-2.75mm}&&\hspace{-2.75mm} = \frac{1}{\varrho_j} \int_{\frac{j+1}{\nu}-x}^{\frac{j+1}{\nu}} f_\text{T}(x) \text{d}x.
\end{eqnarray*}
Thus,  the pdf of $s_j$ can be obtained as
\begin{equation*}
    f_{s_j}(x) = \frac{1}{\varrho_j}f_\text{T}\bigg(\frac{j+1}{\nu}-x\bigg),~~x\in\bigg(0,\frac1\nu\bigg).
\end{equation*}

Taking the expectation over all possible conditions, the pdf of $s$ conditioned on $X_k\leq T_{k-1}$ can be given by
\begin{eqnarray*}
    f_\text{s}(x)\hspace{-2.75mm}&=&\hspace{-2.75mm} \sum_{j=m_0}^\infty \Pr\bigg\{\frac j\nu\leq T_{k-1}<\frac{j+1}{\nu}\bigg| X_k\leq T_{k-1} \bigg\} f_{s_j}(x)  \\
     \hspace{-2.75mm}&=&\hspace{-2.75mm}   \sum_{j=m_0}^\infty \frac{\varrho_j}{\rho_1} \frac{1}{\varrho_j}f_\text{T}\bigg(\frac{j+1}{\nu}-x\bigg)\\
                \hspace{-2.75mm}&=&\hspace{-2.75mm} \frac{\mu(1-\rho_1)w_1}{1-w_1} e^{\mu(1-\rho_1)x}, ~~x\in\bigg(0,\frac1\nu\bigg).
\end{eqnarray*}

For a given $s$, we denote the number of decisions made during $Y_k$ conditioned on $X_k\leq T_{k-1}$ as  $N_k^\text{l}$,  we then  have
    \begin{eqnarray*}
        \Pr\{N_k^\text{l}=0\} \hspace{-2.75mm}&=&\hspace{-2.75mm} \Pr\{ Y_k <  s| X_k\leq T_{k-1} \}\\
            \hspace{-2.75mm}&=&\hspace{-2.75mm} \int_0^{\frac1\nu} f_\text{s}(x) \text{d}x \int_0^x f_\text{S} (y) \text{d}y \\
            \hspace{-2.75mm}&=&\hspace{-2.75mm} 1- \int_0^{\frac1\nu} f_\text{s}(x) e^{-\mu x} \text{d}x  \\
         \hspace{-2.75mm}&=&\hspace{-2.75mm}  1-\frac{(1-\rho_1)(w_1-w_0)}{\rho_1(1-w_1)} \triangleq p_\text{s}, \\
        \Pr\{N_k^\text{l}=j\} \hspace{-2.75mm}&=&\hspace{-2.75mm} \Pr\bigg\{ \frac{j-1}{\nu}+s \leq Y_k < \frac{j}{\nu}+s \bigg|X_k\leq T_{k-1} \bigg\}, \\
        \hspace{-2.75mm}&=&\hspace{-2.75mm} \int_0^{\frac1\nu} f_\text{s}(x) \text{d}x \int_{\frac{j-1}{\nu}+x}^{\frac j\nu+x}f_\text{S}(y) \text{d}y  \\
         \hspace{-2.75mm}&=&\hspace{-2.75mm}  (1-p_\text{s})(1-w_0)w_0^{j-1}~\text{for}~ j=1,2\cdots,
    \end{eqnarray*}
    \begin{equation*}
        \mathbb{E}[N_k^\text{l}] = \frac{1-p_\text{s}}{1-w_0}, \quad
        \mathbb{E}[\big({N_k^{\text{l}}}\big)^2] = \frac{(1-p_\text{s})(1+w_0)}{(1-w_0)^2},
    \end{equation*}
    where $w_0=e^{\frac{-\mu}{\nu}}$.

Note that the AuD of the $i$-th decision made during $Y_k$ can be written as
\begin{equation*}
    \Delta_\text{D}(\tau_{k_i}) = T_{k-1}+s+\frac{i-1}{\nu}.
\end{equation*}
    The expected sum AuD during  $Y_k$ would be
    \begin{eqnarray} \nonumber
        \hspace{-14mm}&&\mathbb{E}[\Delta^\text{l}_{\text{D}k}|X_k\leq T_{k-1}] =\mathbb{E}\Big[ \sum\nolimits_{i=1}^{N_k^\text{l}} \Delta_\text{D}(\tau_{k_i})
                                                                                                            |X_k\leq T_{k-1} \Big] \\
        \label{dr:apx_delta_small}
        \hspace{-14mm}&& = \mathbb{E}\bigg[T_{k-1}+s -\frac{1}{2\nu}\bigg|X_k\leq T_{k-1}\bigg] \mathbb{E}[N_k^\text{l}]  + \frac{1}{2\nu}\mathbb{E}[(N_k^\text{l})^2] \\
        \label{rt:apx_delta_small_0}
        \hspace{-14mm}&&=\frac{1-p_\text{s}}{2\nu(1-w_0)} \Big( \frac{1+w_1}{1-w_1}+ \frac{1+w_0}{1-w_0}+2m_0 \Big).
    \end{eqnarray}
    where \eqref{dr:apx_delta_small} follows the fact that $T_{k-1}+s$ is an integer multiple of $1/\nu$ and is independent of the value of $s$.
        In particular, we have
        \begin{eqnarray*}
            \hspace{-14mm}&& \mathbb{E}[T_{k-1}+s|X_k\leq T_{k-1}]  \\
            \hspace{-14mm}&&  = \sum_{j=0}^\infty \frac{j+1}{\nu} \Pr\bigg\{\frac j\nu\leq T_{k-1}<\frac{j+1}{\nu} \bigg|X_k\leq T_{k-1} \bigg\} \\
            \hspace{-14mm}&& = \sum_{j=0}^\infty \frac{j+1}{\nu} \frac{\varrho_{j}}{\rho_1} \\
            \hspace{-14mm}&& = \frac1\nu\Big(m_0+\frac{1}{1-w_1}\Big).
        \end{eqnarray*}

Secondly, consider the case of $X_k>T_{k-1}$ and $Y_k=X_k-T_{k-1}+S_k$, as shown in Fig. \ref{fig:periodic}(b).
    We denote the parts of $Y_k$ before and after the arrival of update $k$ as $Y_{k1}$ and $Y_{k2}$, respectively, i.e., $Y_{k1}=X_k-T_{k-1}$ and $Y_{k2}=S_k$.
We denote the number of decisions made during $Y_{k1}$ and $Y_{k2}$ as $N_{k1}^\text{g}$ and $N_{k2}^\text{g}$, respectively.
    Since  $Y_{k2}$ follows the exponential distribution, we have
    \begin{eqnarray*}
        \Pr\{N_{k2}^\text{g}=j\} \hspace{-2.75mm}&=&\hspace{-2.75mm} \Pr\bigg\{ \frac{j}{\nu} \leq Y_{k2} <  \frac{j+1}{\nu} \bigg\} \\
         \hspace{-2.74mm}&=&\hspace{-2.75mm}  (1-w_0)w_0^j ~~\text{for} ~~ j=0,1,2\cdots,
    \end{eqnarray*}
    \begin{equation*}
        \mathbb{E}[N_{k2}^\text{g}] = \frac{w_0}{1-w_0}, \quad
        \mathbb{E}[(N_{k2}^{\text{g}})^2] = \frac{w_0+w_0^2}{(1-w_0)^2}.
    \end{equation*}

\begin{figure}[!t]
\centering
\includegraphics[width=1.8in]{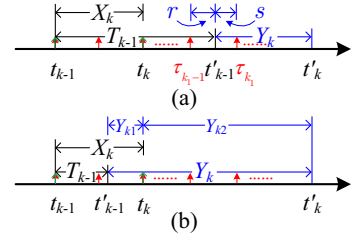}
\caption{Inter-arrival time and system time for synchronous systems. } \label{fig:periodic}
\end{figure}

During $Y_{k1}$, we have
    \begin{eqnarray*}
        \hspace{-2.75mm}&&\hspace{-2.75mm}\Pr\bigg\{N_{k1}^\text{g}=j\}= \Pr\{ \frac{j-1}{\nu} \leq \frac{m_0}{\nu} - T_{k-1} <  \frac{j}{\nu} \bigg| X_k>T_{k-1} \bigg\} \\
         \hspace{-2.74mm}&&\hspace{-2.75mm} =  \frac{1}{1-\rho_1} (1-w_1)w_1^{m_0-j}~\text{for} ~j=1,\cdots,m_0
    \end{eqnarray*}
 and
    \begin{eqnarray*}
        \mathbb{E}[N_{k1}^\text{g}] \hspace{-2.75mm}&=&\hspace{-2.75mm} \frac{m_0}{1-\rho_1} - \frac{w_1}{1-w_1}, \\
         \mathbb{E}\big[\big(N_{k1}^{\text{g}}\big)^2\big] \hspace{-2.75mm}&=&\hspace{-2.75mm} \frac{m_0^2}{1-\rho_1} + \frac{w_1+w_1^2}{(1-w_1)^2} - \frac{2m_0w_1}{(1-\rho_1)(1-w_1)}.
    \end{eqnarray*}

The AuD of decisions made during $Y_{k1}$ and $Y_{k2}$ are given, respectively, by
    \begin{eqnarray*}
        \Delta_{\text{D}1}(\tau_{k_i}) \hspace{-2.75mm}&=&\hspace{-2.75mm} \frac{m_0-i+1}{\nu}, \quad i=1,2,\cdots,  N_{k1},\\
         \Delta_{\text{D}2}(\tau_{k_j}) \hspace{-2.74mm}&=&\hspace{-2.75mm}  \frac{m_0}{\nu} + \frac{j}{\nu}, \quad~\quad j=1,2,\cdots,  N_{k2}.
    \end{eqnarray*}

 The average sum AuD during $Y_k$ on condition $X_k>T_{k-1}$, therefore, is given by,
     \begin{eqnarray} \nonumber
        \hspace{-2.75mm}&& \hspace{-2.75mm}\mathbb{E}[\Delta^\text{g}_{\text{D}k}|X_k>T_{k-1}] =\mathbb{E}\Big[ \sum\nolimits_{i=1}^{N_{k1}^\text{g}} \Delta_{\text{D}1}(\tau_{k_i}) \Big| X_k >T_{k-1}\Big] \\ \nonumber
        \hspace{-2.75mm}&& \hspace{-2.75mm} \qquad\qquad\qquad\qquad~~ + \mathbb{E}\Big[ \sum\nolimits_{j=1}^{N_{k2}^\text{g}} \Delta_{\text{D}2}(\tau_{k_j}) \Big| X_k >T_{k-1}\Big]\\ \nonumber
                    \hspace{-2.75mm}&& \hspace{-2.75mm} =\mathbb{E}\Big[ N_{k1}^\text{g} \frac{m_0+1}{\nu} - \sum\nolimits_{i=1}^{N_{k1}^\text{g}} \frac i\nu \Big]
                                                                                                      +  \mathbb{E}\Big[ N_{k2}^\text{g} \frac{m_0}{\nu} + \sum\nolimits_{j=1}^{N_{k2}^\text{g}} \frac j\nu \Big]  \\ \nonumber
                    \hspace{-2.75mm}&& \hspace{-2.75mm} = \frac{2m_0+1}{2\nu} \mathbb{E}\big[ N_{k1}^\text{g} + N_{k2}^\text{g}\big] +
                                                                                                        \frac{1}{2\nu}\mathbb{E}\Big[ {(N_{k2}^\text{g})}^2-{(N_{k1}^\text{g})}^2\Big] \\ \nonumber
                    \hspace{-2.75mm}&& \hspace{-2.75mm} = \frac{1}{2\nu}\bigg( \frac{w_0+w_0^2}{(1-w_0)^2} -\frac{w_1+w_1^2}{(1-w_1)^2}+\frac{2m_0w_1}{(1-\rho_1)(1-w_1)}\bigg) \\
                    \label{rt:apx_delta_large}
                   \hspace{-2.75mm}&& \hspace{-2.75mm} ~~  +\frac{m_0+m_0^2}{2\nu(1-\rho_1)} + \frac{(2m_0+1)(w_0-w_1)}{2\nu(1-w_0)(1-w_1)}.
    \end{eqnarray}

 Finally, suppose  that $K$  updates are served and $N_\text{T}$ decisions are made during a period $T$, where $K_1$ decisions are made during inter-departure times with $X_k<T_{k-1}$ and $K_2$ decisions are made during inter-departure times with $X_k>T_{k-1}$.
    As $T$ goes to infinity, we have
\begin{eqnarray}\label{rt:fenmu}
    \lim_{T\rightarrow\infty} \frac{N_\text{T}}{K} = \frac{\mathbb{E}[Y_k]}{\frac1\nu} = \nu\mathbb{E}[X_k] = m_0.
\end{eqnarray}
    We would like to mention that the same result can be obtained from $\lim_{T\rightarrow\infty} \frac{N_\text{T}}{K} = \rho_1 \mathbb{E}[N_{k}^\text{l} ]+(1-\rho_1) \mathbb{E}[N_{k1}^\text{g} +N_{k2}^\text{g}  ]$.

Combing the results in \eqref{rt:apx_delta_small_0}, \eqref{rt:apx_delta_large}, and \eqref{rt:fenmu}, we can express the average AuD as
\begin{eqnarray*}
         \hspace{-6mm}&&\widebar{\Delta}_\text{D} = \lim_{T\rightarrow\infty} \frac{1}{N_\text{T}}\sum_{k=1}^K \Delta_{\text{D}k}\\
            \hspace{-6mm}&&=\lim_{T\rightarrow\infty} \frac {K}{N_\text{T}} \left( \frac{K_1}{K} \frac{1}{K_1} \sum_{k=1}^{K_1} \Delta_{\text{D}k}^\text{l} + \frac{K_2}{K} \frac{1}{K_2}\sum_{k=1}^{K_2} \Delta_{\text{D}k}^\text{g}\right) \\
             \hspace{-6mm}&& =\frac{1}{m_0} \Big( \rho_1 \mathbb{E}[ \Delta^\text{l}_{\text{D}k}|X_k<T_{k-1} ] + (1-\rho_1)\mathbb{E}[ \Delta^\text{g}_{\text{D}k}|X_k>T_{k-1}] \Big)\\
             \hspace{-6mm}&&  = \frac{(1-\rho_1)(w_1-w_0)}{2m_0\nu(1-w_0)(1-w_1)}\Big( \frac{1+w_1}{1-w_1}+ \frac{1+w_0}{1-w_0}-1 \Big)  \\
             \hspace{-6mm}&& ~ +\frac{1-\rho_1}{2m_0\nu}\Big( \frac{w_0+w_0^2}{(1-w_0)^2} -\frac{w_1+w_1^2}{(1-w_1)^2}\Big)  + \frac{1+m_0}{2\nu}+\frac{w_1}{\nu(1-w_1)}\\
             \hspace{-6mm}&& = \frac{1+m_0}{2\nu}+\frac{w_1}{\nu(1-w_1)}.
\end{eqnarray*}
Thus, the proof of \textit{Theorem} \ref{thm:thm_1_aud_g_1} is completed.
\end{proof}

\subsection{Proof of Proposition \ref{prop:missing_prob_periodic} } \label{apx:missing_prob_periodic}
\begin{proof}
    Since the event that an update is missed to make any decision is equivalent to the event that the inter-departure time before the update is less than an inter-decision time, we have
    \begin{eqnarray*}
         p_\text{mis} \hspace{-2.75mm}&=&\hspace{-2.75mm} \Pr\left\{ Y_k<\frac{1}{\nu} \right\} \\
            \hspace{-2.75mm}&=&\hspace{-2.75mm} \Pr\{ X_k<T_{k-1}\} \Pr\left\{ Y_k<\frac1\nu \Big| X_k<T_{k-1}\right\} \\
            \hspace{-2.75mm}&&\hspace{-2.75mm} + \Pr\{ X_k>T_{k-1}\} \Pr\left\{ Y_k<\frac1\nu \Big| X_k>T_{k-1}\right\} \\
            \hspace{-2.75mm}&=&\hspace{-2.75mm} \rho_1 \Pr\left\{ S_k<\frac1\nu \right\} \\
            \hspace{-2.75mm}&&\hspace{-2.75mm}  + (1-\rho_1)\Pr\left\{ X_k-T_{k-1}+S_k<\frac1\nu \Big| X_k>T_{k-1}\right\} \\
            \hspace{-2.75mm}&=&\hspace{-2.75mm}  \frac{\rho_1}{2-\rho_1}\left( \frac{1}{w_1}-w_0 \right),
    \end{eqnarray*}
    where $w_0=e^{\frac{-\mu}{\nu}}$, $w_1=e^{\frac{-\mu(1-\rho_1)}{\nu}}$, and $f_\text{T}(x)$ is given by \textit{Proposition} \ref{prop:fx_Tk}.
    This completes the proof of \textit{Proposition} \ref{prop:missing_prob_periodic}.
\end{proof}

\subsection{Proof of Theorem \ref{thm:aud_dm1_periodic_offset}} \label{apx:aud_dm1_periodic_offset}
\begin{proof}
In this proof, notations in Table \ref{tab:notations} are used.

    Firstly, consider the case where the system time is no less than the inter-arrival time, i.e., $X_k\leq T_{k-1}$.
In this case, we have $Y_k=S_k$ and the corresponding probability is $\rho_1$.

Suppose that $T_{k-1}$ consists of $j$ decision intervals and an offset $\delta$, i.e., $\frac j\lambda + \delta< T_{k-1}<\frac{j+1}{\lambda}+ \delta$.
    We denote $r_j=T_{k-1}-\frac j\lambda-\delta$ and $s_j=\frac1\lambda-r_j$.
In particular, we denote $s_0\in(0,\delta)$ as the remaining system time conditioned on $\frac 1\lambda\leq  T_{k-1}<\frac{1}{\lambda}+ \delta$.
    Then we have,
    \begin{eqnarray*}
        \hspace{-2.75mm}&&\hspace{-2.75mm}\Pr\left\{ s_0\leq x\Big|  \frac 1\lambda\leq T_{k-1}<\frac{1}{\lambda}+ \delta \right\}\\
       \hspace{-2.75mm}&&\hspace{-2.75mm}  = \Pr\left\{ T_{k-1}\geq\frac1\lambda+\delta- x \Big|  \frac 1\lambda\leq T_{k-1}<\frac{1}{\lambda}+ \delta \right\}\\
        \hspace{-2.75mm}&&\hspace{-2.75mm} = \frac{1}{\varrho_0} \int_{\frac1\lambda+\delta-x}^{\frac1\lambda+\delta} f_\text{T}(x) \text{d}x,
    \end{eqnarray*}
    where $f_\text{T}(x)$ is given in \textit{Proposition} \ref{prop:fx_Tk}, $u_1=e^{-\mu(1-\rho_1)\delta}$, and
    \begin{eqnarray*}
         \varrho_0= \int_{\frac1\lambda}^{\frac1\lambda+\delta} f_\text{T}(x) \text{d}x
        = \rho_1(1-u_1).
    \end{eqnarray*}

Thus, the pdf of $s_0$ is
\begin{equation*}
    f_{\text{s}_0} (x) = \frac{1}{\varrho_0} f_{\text{T}} \Big(\frac1\lambda+\delta-x\Big), ~~x\in(0,\delta).
\end{equation*}

For $j=1,2,\cdots$, we further denote $\varrho_j=\Pr\{\frac j\lambda+\delta\leq T_{k-1}<\frac{j+1}{\lambda}+\delta\}$ and have
\begin{eqnarray*}
    \varrho_j  =  u_1 (1-\rho_1)\rho_1^j.
\end{eqnarray*}
It is easy to verify that $\sum_{j=0}^\infty \varrho_j=\rho_1$.

Under the condition $\frac j\lambda+\delta \leq T_{k-1}<\frac{j+1}{\lambda}+\delta$, we have
\begin{eqnarray*}
    \hspace{-2.75mm}&&\hspace{-2.75mm}\Pr\left\{s_j\leq x \Big|\frac j\lambda+\delta \leq T_{k-1}<\frac{j+1}{\lambda}+\delta \right\} \\
     \hspace{-2.75mm}&&\hspace{-2.75mm}  = \frac{1}{\varrho_j} \int_{\frac {j+1}{\lambda}+\delta-x}^{\frac {j+1}{\lambda}+\delta} f_\text{T}(x) \text{d}x.
\end{eqnarray*}
Thus,  the pdf of $t_j$ can be obtained as
\begin{equation*}
    f_{s_j}(x) = \frac{1}{\varrho_j}f_\text{T}\left(\frac{j+1}{\lambda}+\delta-x\right)~\text{for}~x\in\left(0,\frac1\lambda\right).
\end{equation*}

Taking the expectation over all possible conditions, the expectation of $s$ can then be obtained as follows
\begin{align*}\label{apx:exp_t}
&\mathbb{E}[s]=\Pr\left\{\frac1\lambda\leq T_{k-1}<\frac1\lambda+\delta\Big|X_k\leq T_{k-1}\right\} \mathbb{E}[s_0] \\
 &~~+\sum_{j=1}^\infty\Pr\left\{\frac j\lambda+\delta\leq T_{k-1}<\frac {j+1}{\lambda}+\delta\Big|X_k\leq T_{k-1}\right\} \mathbb{E}[s_j] \\
&=\frac{\varrho_0}{\rho_1} \int_0^\delta xf_{\text{s}_0}(x)  \text{d}x
                                                +\sum_{j=1}^\infty \frac{\varrho_j}{\rho_1} \int_0^{\frac{1}{\lambda}} xf_{\text{s}_j}(x)  \text{d}x \\
&= \delta-\frac{1}{\mu(1-\rho_1)} +\frac{u_1}{\lambda(1-\rho_1)}.
\end{align*}

Let $N_k^\text{l}$ be the number of decisions made during $Y_k$ conditioned on $X_k\leq T_{k-1}$, we then have
    \begin{eqnarray*}
        \hspace{-3.5mm}&&\hspace{-3.5mm} \Pr\{N_k^\text{l}=0\} \\
        \hspace{-3.5mm}&&\hspace{1.5mm} = \Pr\left\{\frac1\lambda\leq T_{k-1}<\frac1\lambda+\delta \Big|X_k\leq T_{k-1}\right\}  \Pr\{ Y_k <  s_0 \} \\
            \hspace{-3.5mm}&&\hspace{5.5mm} +\sum_{j=1}^\infty\Pr\left\{\frac j\lambda+\delta\leq T_{k-1}
                                                                                                                <\frac {j+1}{\lambda}+\delta\Big|X_k\leq t_{k-1}\right\} \\
            \hspace{-3.5mm}&&\hspace{14.5mm} \cdot \Pr\{ Y_k <  s_j \}\\
            \hspace{-3.5mm}&&\hspace{1.5mm}= \frac{1}{\rho_1} \int_0^{\delta} f_\text{T}\Big(\frac1\lambda +\delta-x\Big) \text{d}x \int_0^x f_\text{S} (y) \text{d}y \\
            \hspace{-3.5mm}&&\hspace{5.5mm} + \frac{1}{\rho_1} \sum_{j=1}^\infty\int_0^{\frac1\lambda}  f_\text{T}\Big(\frac{j+1}{\lambda} +\delta-x\Big) \text{d}x \int_0^x f_\text{S} (y) \text{d}y  \\
         \hspace{-3.5mm}&&\hspace{1.5mm}=  1- \frac{1}{\rho_1} ((1-\rho_0)u_1-(1-\rho_1)u_0) \triangleq p_\text{s}, \\
  \hspace{-3.5mm}&&\hspace{-3.5mm} \Pr\{N_k^\text{l}=j\} =\Pr\left\{ \frac{j-1}{\lambda}+s \leq Y_k <  \frac{j}{\lambda}+s \right\}, \\
          \hspace{-3.5mm}&&\hspace{1.5mm} =  (1-p_\text{s})(1-\rho_0)\rho_0^{j-1}~\text{for}~ j=1,2\cdots,
    \end{eqnarray*}
and
    \begin{equation*}
        \mathbb{E}[N_k^\text{l}] = \frac{1-p_\text{s}}{1-\rho_0}, \quad
        \mathbb{E}\big[\big(N_k^{\text{l}}\big)^2\big] = \frac{(1-p_\text{s})(1+\rho_0)}{(1-\rho_0)^2}.
    \end{equation*}

Since the AuD of the $i$-th decision made during $Y_k$ is
\begin{equation*}
    \Delta_\text{D}(\tau_{k_i}) = T_{k-1}+s+\frac{i-1}{\lambda},
\end{equation*}
    the expected sum AuD during $Y_k$ would be
    \begin{eqnarray*} \nonumber
        \hspace{-8mm}&& \mathbb{E}[\Delta^\text{l}_{\text{D}k}|X_k\leq T_{k-1}] =\mathbb{E}\left[ \sum\nolimits_{i=1}^{N_k^\text{l}} \Delta_\text{D}(\tau_{k_i})\Big|X_k\leq T_{k-1}  \right]\\ \nonumber
        \hspace{-8mm}&&  =\left(\mathbb{E}[T_{k-1}|X_k\leq T_{k-1}] + \mathbb{E}(s) -\frac{1}{2\lambda} \right)
                \mathbb{E} [N_{k}^\text{l}] + \frac{1}{2\lambda} \mathbb{E} \big[(N_{k}^\text{l})^2\big]\\
             \label{rt:apx_delta_small}
        \hspace{-8mm}&& = \frac{1-p_\text{s}}{1-\rho_0} \left(\delta+\frac{u_1}{\lambda(1-\rho_1)}+\frac{1}{\lambda(1-\rho_0)}\right).
    \end{eqnarray*}

Secondly, consider the case of $X_k>T_{k-1}$ and $Y_k=X_k-T_{k-1}+S_k$, as shown in Fig. \ref{fig:periodic_offset}(b).
    We denote the parts of $Y_k$ before and after the arrival of update $k$ as $Y_{k1}$ and $Y_{k2}$, respectively.
We denote the number of decisions made during $Y_{k1}$ and $Y_{k2}$ as $N_{k1}^\text{g}$ and $N_{k2}^\text{g}$, respectively.
    Since  $Y_{k2}=S_k$ follows the exponential distribution, we have
    \begin{eqnarray*}
        \Pr\{N_{k2}^\text{g}=0\} \hspace{-2.75mm}&=&\hspace{-2.75mm} \Pr\{ Y_{k2} <  \delta \} \\
         \hspace{-2.74mm}&=&\hspace{-2.75mm}  1-u_0, \\
        \Pr\{N_{k2}^\text{g}=j\} \hspace{-2.75mm}&=&\hspace{-2.75mm} \Pr\left\{ \frac{j-1}{\lambda} +\delta\leq Y_{k2} <  \frac{j}{\lambda} +\delta \right\}, \\
         \hspace{-2.74mm}&=&\hspace{-2.75mm}  u_0(1-\rho_0)\rho_0^{j-1} ~\text{for} ~ j=1,2,\cdots,
    \end{eqnarray*}
    and
    \begin{equation*}
        \mathbb{E}[N_{k2}^\text{g}] = \frac{u_0}{1-\rho_0}, \quad
        \mathbb{E}\big[(N_{k2}^{\text{g}})^2\big] = \frac{u_0(1+\rho_0)}{(1-\rho_0)^2},
    \end{equation*}
    where $u_0=e^{-\mu\delta}$.

\begin{figure}[!t]
\centering
\includegraphics[width=1.8in]{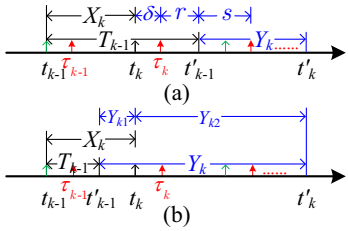}
\caption{Inter-arrival time and system time  for asynchronous systems. } \label{fig:periodic_offset}
\end{figure}

During $Y_{k1}=X_k-T_{k-1}$, there is at most one decision with probability
    \begin{equation*}
        \Pr\{N_{k1}^\text{g}=1\}= \Pr\{ T_{k-1} <  \delta \}=\frac{1-u_0}{1-\rho_1},
    \end{equation*}
Thus we  have
    \begin{equation*}
        \mathbb{E}[N_{k1}^\text{g}] = \frac{1-u_0}{1-\rho_1},~~
         \mathbb{E}\big[{(N_{k1}^{\text{g}})}^2\big] = \frac{1-u_0}{1-\rho_1}.
    \end{equation*}

The AuD of decisions made during $Y_{k1}$ and $Y_{k2}$ are given, respectively, by
    \begin{eqnarray*}
        \Delta_\text{D}(\tau_{k_i}) \hspace{-2.75mm}&=&\hspace{-2.75mm} \delta, \qquad\quad i=1,\\
         \Delta_\text{D}(\tau_{k_i}) \hspace{-2.74mm}&=&\hspace{-2.75mm}   \frac{i}{\lambda}+\delta,\quad i=1,2,\cdots,  N_{k2}^\text{g}.
    \end{eqnarray*}

 The average sum AuD during $Y_k$ under the condition $X_k>T_{k-1}$, therefore, is given by
     \begin{eqnarray*}
        \hspace{-8mm}&& \mathbb{E}[\Delta_{\text{D}k}|X_k>T_{k-1}] \\
        \hspace{-8mm}&&  =\Pr\{N_{k1}^\text{g}=1\} \Delta_\text{D}(\tau_{k_1}) + \mathbb{E}\Big[ \sum\nolimits_{i=1}^{N_{k2}^\text{g}} \Delta_\text{D}(\tau_{k_i}) \Big |X_k>T_{k-1}\Big]\\
         \hspace{-8mm}&&  = \Pr\{N_{k1}^\text{g}=1\} \delta+ \left(\delta+\frac{1}{2\lambda}\right) \mathbb{E}[ N_{k2}^\text{g}] + \frac{1}{2\lambda} \mathbb{E}\big[ (N_{k2}^\text{g})^2\big] \\
      \hspace{-8mm}&& = \frac{(1-u_1)\delta}{1-\rho_1} + \frac{u_0\delta}{1-\rho_0} + \frac{u_0}{\lambda(1-\rho_0)^2}.
    \end{eqnarray*}

 Finally, suppose  that $K$  updates are served and $N_\text{T}$ decisions are made during a period $T$, where $K_1$ decisions are made during inter-departure times with $X_k\leq T_{k-1}$ and $K_2$ decisions are made during inter-departure times with $X_k>T_{k-1}$.
    As $T$ goes to infinity, we have $\lim_{T\rightarrow\infty} \frac{N_\text{T}}{K} =1$ and
\begin{eqnarray*}
         \widebar{\Delta}_\text{D} \hspace{-3mm}&=&\hspace{-2.75mm} \lim_{T\rightarrow\infty} \frac{1}{N_\text{T}}\sum_{k=1}^K \Delta_{\text{D}k}\\
                                \hspace{-3mm}&=&\hspace{-2.75mm}  \lim_{T\rightarrow\infty} \frac {K}{N_\text{T}} \left( \frac{K_1}{K} \frac{1}{K_1} \sum_{k=1}^{K_1} \Delta_{\text{D}k}^\text{l}+ \frac{K_2}{K} \frac{1}{K_2}\sum_{k=1}^{K_2} \Delta_{\text{D}k}^\text{g}\right) \\
             \hspace{-3mm}&=&\hspace{-2.75mm} \rho_1 \mathbb{E}[ \Delta_{\text{D}k}|X_k\leq T_{k-1} ]  + (1-\rho_1) \mathbb{E}[ \Delta_{\text{D}k}|X_k>T_{k-1}] \\
             \hspace{-3mm}&=&\hspace{-2.75mm} \delta + \frac{(1-\rho_0)u_1^2+(1-\rho_1)(1-u_0)u_1}{\lambda(1-\rho_1)(1-\rho_0)}.
\end{eqnarray*}
Thus, the proof of \textit{Theorem} \ref{thm:aud_dm1_periodic_offset} is completed.

\end{proof}

\ifCLASSOPTIONcaptionsoff
  \newpage
\fi

\small{
\bibliographystyle{IEEEtran}

\begin{thebibliography}{11}
\bibitem{Monitoring-IoTJ-2018}
F, Montori, L. Bedogni, and L. Bononi, ``A collaborative Internet of Things architecture for
smart cities and environmental monitoring," \textit{IEEE Internet Things J.}, vol. 5, no. 2, pp. 592--605, Apl. 2018.

\bibitem{Mobile-IoTJ-2016}
A. Kamilaris and A. Pitsillides, ``Mobile phone computing and the Internet of Things: A survey," \textit{IEEE Internet Things J.}, vol. 3, no. 6, pp. 885--897, Dec. 2016.

\bibitem{IIoT-IoTJ-2018}
B. Nguyen, T.-D. Hoa, and D.-S. Kim, ``Energy-aware real-time routing for large-scale industrial Internet of Things," \textit{IEEE Internet Things J.}, vol. 5, no. 3, pp. 2190--2199, Dec. 2018.


\bibitem{Vnet-1-2011}
S. Kaul, M. Gruteser, V. Rai, and J. Kenney, ``Minimizing age of information in vehicular networks," in \textit{Proc. IEEE SECON}, Salt Lake, Utah, USA, Jun. 2011, pp. 350--358.



\bibitem{health-Proc-2012}
I. Bisio, F. Lavagetto, M. Marchese, and A. Sciarrone, ``Smartphone-based user activity recognition method for health remote monitoring applications", in \textit{ Proc. Intl. Conf.  Pervasive  Embedded
Comput. Commun. Sys.}, Rome, Italy, Feb. 2012, pp. 200--205.

\bibitem{asset-PMC-2016}
I. Bisio, A. Sciarrone, and S. Zappatore, ``A new asset tracking architecture integrating RFID, bluetooth low energy tags and ad hoc smartphone applications" , \textit{Elsevier Pervasive Mobile Comput.},  vol. 31, no.1, pp. 79--93, Jan. 2016.

\bibitem{remote-Proc-2017}
I. Bisio, F. Lavagetto, M. Marchese, and A. Sciarrone, ``Smart probabilistic fingerprinting for WiFi-based indoor positioning with mobile devices," \textit{Elsevier Pervasive Mobile Comput.}, vol. 31, no. 2, pp. 107--123, Feb. 2016.


\bibitem{tutorial-2016}
A. Kosta, N. Pappas, and V. Angelakis, ``Age of information: A new concept, metric, and tool," \textit{Found. Trends Netw.}, vol. 12, no. 3, pp. 162--259, Nov. 2017.

\bibitem{Yates-2012-age}
S. K. Kaul, R. D. Yates, and M. Gruteser, ``Real-time status: How often should one update?" in \textit{Proc. IEEE INFOCOM}, Orlando, FL, USA, Mar. 2012, pp. 2731--2735.

\bibitem{Vnet-2-2011}
S. Kaul, R. Yates, and M. Gruteser, ``On piggybacking in vehicular networks,” in \textit{Proc. IEEE Global Commun. (GLOBECOM)}, Houston, Texas, USA, Dec. 2011, pp. 1--5.

\bibitem{SHS-2018}
R. Yates, ``The age of information in networks: Moments, distributions, and sampling," [Online]. Available: arXiv:1806.03487v1.

%
%
%

\bibitem{SmartAgriculture-2016}
Map Your Meal Europe Aid Funded Project. (2016). [Online].
Available: \url{http://www.mapyourmeal.org/}

\bibitem{Sun-2016-mlt-sver}
A. M. Bedewy, Y. Sun, and N. B. Shroff, ``Optimizing data freshness, throughput, and delay in multi-server information-update systems," \textit{Proc. IEEE Int. Symp. Inf. Theory (ISIT)}, Barcelona, Spain, Jul. 2016, pp.  2569--2573.



\bibitem{Inoue-2017-distribution}
Y. Inoue, H. Masuyama, T. Takine, and T. Tanaka, ``The stationary distribution of the age of information in FCFS single-server queues," in \textit{Proc. IEEE Int. Symp. Inf. Theory (ISIT)}, Aachen, Germany, Jun. 2017, pp. 571--575.

\bibitem{Hr-2019-correlated}
J. Hribar, M. Costa, N. Kaminski, and L. A. DaSilva, ``Using correlated information to extend device lifetime," \textit{IEEE Internet Things J.}, vol. 6, no. 2, pp. 2439--2448, Apl. 2019.

\bibitem{Gu-20197-mornitoring}
Y. Gu, H. Chen, Y. Zhou, Y. Li,  and B. Vucetic, ``Timely status update in Internet of Things monitoring systems: An age-energy tradeoff," \textit{IEEE Internet Things J.}, vol. 6, no. 3, pp. 5324--5335, Jun. 2019.

\bibitem{Niu-2019-RR1}
Z. Jiang, B. Krishnamachari, X.  Zheng, S. Zhou, and Z. Niu,``Timely status update in wireless uplinks: Analytical solutions with asymptotic optimality," \textit{IEEE Internet Things J.}, vol. 6, no. 2, pp. 3885--3898, Apl. 2019.

\bibitem{Gu-2019-cognitive}
Y. Gu, H. Chen, C. Zhai, Y.  Li, and B. Vucetic, ``Minimizing age of information in cognitive radio-based IoT systems: Underlay or overlay?" \textit{IEEE Internet Things J.}, vol. 99, no. 5, pp. 3885--3898, Oct. 2019.

\bibitem{Dong-2019-access}
Y. Dong, Z. Chen, and P.i Fan, ``Timely two-way data exchanging in unilaterally powered fog computing systems," \textit{IEEE Access}, vol. 7, pp. 21103--21117, Feb. 2019.

\bibitem{Dong-2019-jcn}
C. Hu and Y. Dong, ``Age of information of two-way data exchanging system with power-splitting," \textit{IEEE J. Commun. Netw.}, vol. 21, no. 3, pp. 295--306, Jun. 2019.

\bibitem{kosta-AoI-2017}
A. Kosta, N. Pappas, A. Ephremides, V. Angelakis, ``Age and value of information: Non-linear age case," in \textit{Proc. IEEE Int. Symp. Inf. Theory (ISIT)}, Aachen, Germany, Jun. 2017, pp. 326--330.

\bibitem{Effective-AoI-2018}
C. Kam, S.  Kompella, G.  D. Nguyen, J.  E. Wieselthier, A. Ephremides, ``Towards an “effective age” concept," in \textit{Proc. Int. Wkshp  Signal Process. Adv. Wireless Commun. (SPAWC)}, Kalamata, Greece, Jun. 2018, pp. 1--5.

\bibitem{Mutual-AoI-2018}
Y. Sun, B.  Cyr, ``Information aging through queues: A mutual information perspective," in \textit{Proc. Int. Wkshp  Signal Process. Adv. Wireless Commun. (SPAWC)}, Kalamata, Greece, Jun. 2018, pp. 1--5.

\bibitem{Dong-AuD-2018}
Y. Dong, Z. Chen, S. Liu, and P. Fan, ``Age of information upon decisions," in \textit{Proc. IEEE Sarnoff Symp.}, Newark, NJ, USA, Sep. 2018, pp. 1--6.

\bibitem{Dong-AuD-G-2018}
Y. Dong and P. Fan, ``Age upon decisions with general arrivals," in \textit{Proc. IEEE Annu. Ubiquitous Comput. Electron. Mobile Commun. Conf. (UEMCON)}, New York, USA, Nov. 2018, pp. 825--829.

%
%








\bibitem{Sun-2016-zero-wait}
Y. Sun, E. U. Biyikoglu, R. D. Yates, C. E. Koksal and N. B. Shroff,  ``Update or wait: How to keep your data fresh,"  \textit{IEEE Trans. Inf. Theory}, vol. 63, no. 11, pp. 7492--7508, Nov. 2017.
%
%
%
%
%
%



%



\bibitem{Boyd-2004-convex}
S. Boyd and L. Vandenberghe, \textit{Convex Optimizations}. Cambrige, U.K.: Cambrige Univ. Press, 2004.

%
%

\bibitem{book_queueing_2006}
W. J. Stewart, \textit{Probability, Markov chains, queues, and simulation: the mathematical basis of performance modeling}. Princeton University Press, 2009.


%
\end{thebibliography}

}

\end{document}